\documentclass[conference]{IEEEtran}
\usepackage{amssymb,amsmath,amsthm}
\usepackage{mathrsfs} 

\DeclareFontFamily{OT1}{pzc}{}
\DeclareFontShape{OT1}{pzc}{m}{it}{2 <-> pzcmi8t}{}
\DeclareFontShape{OT1}{pzc}{m}{it}{<-> [1.5] pzcmi8t}{}
\DeclareMathAlphabet{\mathpzc}{OT1}{pzc}{m}{it}

\usepackage{MnSymbol}
\usepackage[latin1]{inputenc}
\usepackage[english]{babel}
\usepackage[]{subfig}

\usepackage[]{color}

\usepackage{epic,eepic,epsfig,graphics}

\usepackage{cite}

\usepackage{color}						
\definecolor{red}{rgb}{1,0,0}	

\newtheorem{theorem}{Theorem}[section]
\newtheorem{lemma}[theorem]{Lemma}

\newtheorem{definition}[theorem]{Definition}

\newcommand{\beq}{\begin{equation}}
\newcommand{\eeq}{\end{equation}}
\newcommand{\bea}{\begin{eqnarray}}
\newcommand{\eea}{\end{eqnarray}}

\def\qed{\quad \vrule height6.5pt width6pt depth0pt} 

\newcommand{\be}{\begin{equation}}
\newcommand{\ee}{\end{equation}}
\newcommand{\beqr}{\begin{eqnarray}}
\newcommand{\eeqr}{\end{eqnarray}}
\newcommand{\beqrx}{\begin{eqnarray*}}
\newcommand{\eeqrx}{\end{eqnarray*}}
\newcommand{\ba}{\left[ \begin{array}}
\newcommand{\ea}{\\ \end{array} \right]}
\newcommand{\bi}{\begin{itemize}}
\newcommand{\ei}{\end{itemize}}

\newcommand{\qd}{\hfill{\qed}}
\def\qed{{\ \vrule width 2.5mm height 2.5mm \smallskip}}

\def\xb{{\bf x}}
\def\sb{{\bf s}}
\def\yb{{\bf y}}
\def\zb{{\bf z}}

\def\vb{{\bf v}}
\def\db{{\bf d}}
\def\hb{{\bf h}}
\def\upsib{\boldsymbol{\upsilon}}
\def\delb{\boldsymbol{\delta}}

\def\Hb{{\bf H}}
\def\Ib{{\bf I}}

\def\Rbb{\mathbb{R}}
\def\Nc{\mathcal{N}}

\def\Ec{\mathcal{E}}

\def\SNR{\mbox{SNR}}

\def\oneb{\text{\boldmath $1$}}
\def\zerob{\text{\boldmath $0$}}
 
\usepackage{graphicx}
\usepackage[linesnumberedhidden,ruled,vlined]{algorithm2e}
\IEEEoverridecommandlockouts

\begin{document}
\title{Optimized Markov Chain Monte Carlo for Signal Detection in MIMO Systems: an Analysis of Stationary Distribution and Mixing Time}

\author{\IEEEauthorblockN{
Babak Hassibi \IEEEauthorrefmark{2}, Morten Hansen\IEEEauthorrefmark{1},  Alexandros Georgios Dimakis \IEEEauthorrefmark{3}, Haider Ali Jasim Alshamary \IEEEauthorrefmark{4}, and Weiyu Xu \IEEEauthorrefmark{4} }
%
%
\thanks{Part of this paper was presented in the IEEE Global Communications Conference 2009 \cite{Hassibi_Globecom}, and the 51st IEEE Conference on Decision and Control 2012  \cite{WeiyuCDC2012}.

H. Alshamary and W. Xu are with the Department of Electrical and Computer Engineering, the University of Iowa, Iowa City, IA, USA. Email: haider-alshamary@uiowa.edu and weiyu-xu@uiowa.edu.
A. Dimakis  is with the Department of Electrical and Computer Engineering, University of Texas, Austin, Austin, TX, USA. Email: dimakis@austin.utexas.edu.
M. Hansen is with  Renesas Mobile Corporation, Sluseholmen 1, 2450 Copenhagen SV, Denmark. E-mail: morten.hansen@renesasmobile.com.
B. Hassibi is with the Department of Electrical Engineering, California Institute of Technology, Pasadena, CA 91125, USA. Email: hassibi@systems.caltech.edu.
}
}
\maketitle

\begin{abstract}
In this paper we introduce an optimized Markov Chain Monte Carlo (MCMC) technique for solving the integer least-squares (ILS) problems, which include Maximum Likelihood (ML) detection in Multiple-Input Multiple-Output (MIMO) systems. Two factors contribute to the speed of finding the optimal solution by the MCMC detector: the probability of the optimal solution in the stationary distribution, and the mixing time of the MCMC detector. Firstly, we compute the optimal value of the ``temperature" parameter, in the sense that the temperature has the desirable property that once the Markov chain has mixed to its stationary distribution, there is polynomially small probability ($1/\mbox{poly}(N)$, instead of exponentially small) of encountering the optimal solution. This temperature is shown to be at most $O(\sqrt{\SNR}/\ln(N))$ \footnote{In this paper, $\Omega(\cdot)$,  $\Theta(\cdot)$, and $O(\cdot)$ are the usual scaling notations as in computer science}, where $\SNR$ is the signal-to-noise ratio, and $N$ is the problem dimension. Secondly, we study the mixing time of the underlying Markov chain of the proposed MCMC detector. We find that, the mixing time of MCMC is closely related to whether there is a local minimum in the lattice structures of ILS problems. For some lattices without local minima, the mixing time of the Markov chain is independent of $\SNR$,  and grows polynomially in the problem dimension; for lattices with local minima, the mixing time grows unboundedly as $\SNR$ grows, when the temperature is set, as in conventional wisdom, to be the standard deviation of noises. Our results suggest that, to ensure fast mixing for a fixed dimension $N$, the temperature for MCMC should instead be set as $\Omega(\sqrt{\SNR})$ in general. Simulation results show that the optimized MCMC detector efficiently achieves approximately ML detection in  MIMO systems having a huge number of transmit and receive dimensions.
\end{abstract}

\section{Introduction}
\label{sec:Introduction}
The problem of performing Maximum Likelihood (ML) decoding in digital communication has gained much attention over the years. These ML decoding problems often reduce to integer least-squares (ILS) problems, which aim to find an integer lattice point closest to received signals. In fact, the ILS problem is an NP-hard optimization problem appearing in many research areas, for example, communications, global navigation satellite systems, radar imaging, Monte Carlo second-moment estimation, bioinformatics, and lattice design \cite{Agrell_et_al_02, Borno}. A computationally efficient way of exactly solving the ILS problem is the sphere decoder (SD) \cite{Damen_et_al, Hochwald_Ten-Brink_03, Hassibi_1,Hassibi_2, Agrell_et_al_02}. It is known that for a moderate problem size and a suitable range of Signal-to-Noise Ratios ($\SNR$), SD has low computational complexity, which can be significantly smaller than an exhaustive search solver. But for a large problem size and fixed $\SNR$, the average computational complexity of SD is still exponential in the problem dimension \cite{Ottersten_05}. So for large problem sizes, (for example massive Multiple-Input Multiple-Output (MIMO) systems with many transmit and receive antennas \cite{LargeAntenna}\cite{TabuSearchLargeAntenna}), SD still has high computational complexity and is thus computationally infeasible. A way to overcome this problem is to use approximate Markov Chain Monte Carlo (MCMC) detectors instead, which can provide the optimal solution asymptotically  \cite{Robert_MC_Stat_methods_04, Haggstrom_02}.

Unlike SD, MCMC algorithms perform a random walk over the signal space in the hope of finding the optimal solution. Glauber dynamics is a popular MCMC method which performs the random walk according to the transition probability determined by the stationary distribution of a reversible Markov chain \cite{Haggstrom_02}\cite{Levin}. \cite{Wang_Poor_03, Zhu_Farhang_Boroujeny_05} proposed Glauber dynamics MCMC detectors for data detection in wireless communication (see also the references therein). These MCMC methods are able to provide the optimal solution if they are run for a sufficiently long time; and empirically MCMC methods are observed to provide near-optimal solutions in a reasonable amount of computational time even for large problem dimensions \cite{Wang_Poor_03, Zhu_Farhang_Boroujeny_05}. However, as observed in \cite{Zhu_Farhang_Boroujeny_05, Farhang_Boroujeny_06, ChenRong}, unlike sphere decoders , which performs well in high SNR regimes, MCMC detectors often suffer performance degradation in high SNR regimes. In fact, the MCMC detectors in the literature were proposed mostly as practical heuristic detectors for digital communications, and theoretical understandings of the performance and complexity of MCMC detectors for ILS problems are very limited. For example, the mixing time (convergence rate) of the underlying Markov chains of these MCMC detectors, namely how fast these Markov chains mix to its stationary distribution, is not explicitly known. For the MCMC detectors in the literature \cite{Wang_Poor_03, Zhu_Farhang_Boroujeny_05}, the conditional transition probabilities of their underlying Markov chains were directly determined by the posterior likelihood of signal sequences \cite{Wang_Poor_03, Zhu_Farhang_Boroujeny_05}. In other words, the standard deviation of channel noise was naturally applied as the ``temperature'' of these MCMC detectors \cite{Wang_Poor_03, Zhu_Farhang_Boroujeny_05}.  However, it was not clear whether this choice of temperature is optimal, and what effect it will have on the performance and complexity of MCMC detectors.
%


Two factors contribute to the speed of finding the optimal solution by the MCMC detector: the probability of the optimal solution in the stationary distribution, and the mixing time of the underlying Markov chain for the MCMC detector. In fact, if the optimal solution has a high probability in the stationary distribution, the MCMC detector will very likely encounter the optimal solution when its underlying Markov chain has mixed to its stationary distribution. However, as we will see in this paper, increasing the probability in the stationary distribution of the optimal solution often (even though not always) results in a slow mixing of the underlying Markov chain, namely it takes long time for the Markov chain to reach its stationary distribution. How to balance the mixing time and the stationary distribution for best performance of MCMC detectors is the main subject of this paper.

Our main contributions in this paper are twofold: characterizing the stationary distribution, and bounding the mixing time of MCMC detectors. These results lead to an optimized MCMC detector for solving ILS problems. Firstly, we compute the optimal value of the ``temperature" parameter, in the sense that the temperature has the desirable property that once the Markov chain has mixed to its stationary distribution, there is polynomially (and not exponentially) small probability of encountering the optimal solution. This temperature is shown to be at most $O(\sqrt{\SNR}/\ln(N))$, where $\SNR$ is the signal-to-noise ratio, and $N$ is the problem dimension. Secondly, we study the mixing time of the underlying Markov chain of the proposed MCMC detector. We find that, the mixing time of MCMC is closely related to whether there is a local minimum in the lattice structures of ILS problems. For some lattices without local minima, the mixing time of the Markov chain is independent of $\SNR$,  and grows polynomially in the problem dimension; for lattices with local minima, the mixing time grows unboundedly as $\SNR$ grows, when the temperature is set, as in conventional wisdom, to be the standard deviation of noises. We also study the probability that there exist local minima in an ILS problem. For example, the probability of having local minima is $\frac{1}{3}-\frac{1}{\sqrt{5}}+\frac{2\arctan(\sqrt{\frac{5}{3}})}{\sqrt{5}\pi}$ for $2 \times 2$ Gaussian MIMO matrices. Simulation results indicate, when the system dimension $N\rightarrow \infty$, there seems to be always at least one local minimum with Gaussian MIMO matrices, but we do not have a rigorous proof of this phenomenon. Our theoretical and empirical results suggest that, to ensure fast mixing, for a fixed dimension $N$, the temperature for MCMC should instead be set as $\Omega(\sqrt{\SNR})$, contrary to conventional wisdom of using the standard deviation of channel noises \cite{Wang_Poor_03, Zhu_Farhang_Boroujeny_05}. Our simulation results show that the optimized MCMC detector efficiently achieves approximately ML detection in  MIMO systems having a huge number of transmit and receive dimensions.

We caution, however, that we have not been able to prove the scaling of the mixing time in terms of system dimension $N$ for integer least-squares problems. The question whether MCMC detectors mix in polynomial time over $N$ remains an open problem.

The paper is organized as follows. In Section \ref{sec:System_model} we present the system model that will be used throughout the paper. The MCMC methods and background knowledge on Markov chain mixing time  are described in Section \ref{sec:MCMC_sampling}. In Section \ref{sec:P_error} we analyze the probability of error for the ML detector. Section \ref{sec:comp_opt_alpha} treats the optimal selection of the temperature parameter $\alpha$. Section \ref{sec:mixing_time_ortho}, \ref{sec:mixing_time_local},\ref{sec:local_minimum} and \ref{sec:choice_alpha} derive bounds on the mixing time and discuss how to optimize MCMC parameters to ensure fast mixing. Simulation results are given in Section \ref{sec:sim_results}.
\section{System Model}
\label{sec:System_model}
We consider a real-valued block-fading MIMO antenna system, with $N$ transmit and $N$ receive dimensions, with know channel coefficients. The received signal $\yb \in \Rbb^N$ can be expressed as
\begin{equation}
\label{EQ:sig_model_matrix}
\yb = \sqrt{\frac{\SNR}{N}}{\Hb{\xb}} + \upsib \ ,
\end{equation}
where ${\xb} \in \Xi^{N}$ is the transmitted signal, and $\Xi$ denotes the constellation set. To simplify the derivations in the paper we will assume that $\Xi = \left\{\pm 1\right\}$. $\upsib \in \Rbb^{N}$ is the noise vector where each entry is Gaussian $\Nc \left(0,1\right)$ and independent identically distributed (i.i.d.), and $\Hb \in \Rbb^{N \times N}$ denotes the channel matrix with i.i.d. $\Nc \left(0,1\right)$ entries. (In general, $\Hb$ can be any matrix, however, for analysis purposes we will focus on $\Hb$ with  i.i.d. Gaussian elements.)
The signal-to-noise ratio is defined as
\begin{equation}
\label{EQ:SNR}
\begin{split}
\SNR &= \frac{\Ec \left\|\sqrt{\frac{\SNR}{N}}\Hb{\xb}\right\|^2}{\Ec \|\upsib\|^2} \ ,
\end{split}
\end{equation}
which is done in order to take into account the total transmit energy. For analysis purposes we will focus on the regime where $\SNR > 2\ln(N)$, in order to get the probability of error of the ML detector to go to zero. Without loss of generality, we will assume that the all-minus-one vector was transmitted, ${\xb} = -\oneb$. Therefore
\begin{equation}
\yb = \upsib-\sqrt{\frac{\SNR}{N}}\Hb\oneb \ .
\end{equation}

We are considering a minimization of the average error probability $P\left({\bf e} \right) \triangleq P\left(\hat{{\xb}} \neq {\xb} \right)$, which is obtained by performing Maximum Likelihood Sequence Detection (here simply referred to as ML detection) given by
\begin{equation}
\label{EQ:Original_minimization_problem}
 \hat{{\xb}} = \arg	\mathop {\text{min} }\limits_{{\xb} \in \Xi^{N}} \ \  \left\| \yb - \sqrt{\frac{\SNR}{N}}\Hb {\xb} \right\|^2 \ .
\end{equation}

\section{MCMC Detector}
\label{sec:MCMC_sampling}

One way of solving the optimization problem given in \eqref{EQ:Original_minimization_problem} is by using Markov Chain Monte Carlo (MCMC) detectors, which asymptotically converge to the optimal solution if the detector follows a reversible Markov chain \cite{Mackay_03}.  We first describe our proposed MCMC detector based on reversible Markov chain, and then compare it with existing MCMC detectors in the literature.

\subsection{Reversible MCMC Detector}
 \label{subsection:reversible}
 In this paper, we mainly focus on a MCMC detector which follows a reversible Markov chain and asymptotically converges to the stationary distribution \cite{Mackay_03}. Under the stationary distribution, the MCMC detector has a certain probability of visiting the optimal solution. So if run for a sufficiently long time, the MCMC detector will be able to find the optimal solution to \eqref{EQ:Original_minimization_problem}.

For this MIMO detection problem \eqref{EQ:Original_minimization_problem}, the MCMC detector starts with a certain $N$-dimensional feasible vector $\hat{\xb}^{(0)}$ among the set $\{-1,+1\}^{N}$ of cardinality $2^{N}$. Then the MCMC detector performs a random walk over $\{-1,+1\}^{N}$ based on the following reversible Markov chain. Assume that we are at time index $l$ and the current state of the Markov chain is $\hat{\xb}^{(l)} \in \{-1,+1\}^{N}$.  In the next step, the Markov chain picks one random position index $j$ uniformly out of $\{1,2, ..., N\}$, and keeps the symbols of $\hat{\xb}^{(l)}$ at other positions fixed. Then the MCMC detector computes the conditional probability of transferring to each constellation point at the $j$-th index. With the symbols at the $(N-1)$ other positions fixed, the probability that the $j$-th symbol adopts the value $\omega$, is given by
\begin{equation}
\label{Eq:Prob_of_symbol_MCMC}
p\left( {\hat{\xb}_j^{(l+1)} = \omega \left|{ \theta }\right.} \right) =  \frac{e^{-\frac{1}{2\alpha^2} \left\| \yb - \sqrt{\frac{\SNR}{N}} \Hb \hat{\xb}_{j \left|{\omega}\right.} \right\|^2   }}{ \sum\limits_{\hat{\xb}_{j \left|{\tilde{\omega}}\right.} \in \Xi}{e^{-\frac{1}{2\alpha^2} \left\| \yb - \sqrt{\frac{\SNR}{N}} \Hb \hat{\xb}_{j \left|{ \tilde{\omega} }\right.} \right\|^2 } }}  \ ,
\end{equation}
where $\hat{\xb}_{j \left|{\omega}\right. }^T \triangleq \left[\hat{\xb}_{1:j-1}^{(l)}, \omega, \hat{\xb}_{j+1:N}^{(l)} \right]^T$ and $\theta = \left\{ \hat{\xb}^{(l)}, j, \yb, \Hb \right\}$. So conditioned on the $j$-th position is chosen, the MCMC detector will with probability $p\left( {\hat{\xb}_j^{(l+1)} = \omega \left|{\theta}\right.} \right)$ transition to $\omega$ at the $j$-th position index. The initialization of the symbol vector $\hat{\xb}^{(0)}$ can either be chosen randomly or as other heuristic solutions.

Our algorithm is summarized as follows in Algorithm \ref{alg:reversibleMCMC}.
\begin{algorithm}
\LinesNumbered
  \caption{MCMC detector based on reversible Markov chain}

    \KwIn{$\yb$, $\Hb$, initialization vector $\hat{\xb}^{(0)}$, decision vector $\hat{\xb}=\hat{\xb}^{(0)}$ and the number of iterations $n$}
    \For {$i=1$ \KwTo $n$} {
        pick a uniformly random position index $j$ out of $\{1,2, ..., N\}$  \par
        keep the symbols of $\hat{\xb}^{(i-1)}$ at the $(N-1)$ other positions fixed, transition the $j$-th symbol of $\hat{\xb}^{(i-1)}$  to $\omega$ with probability $p\left( {\hat{\xb}_j^{(l+1)} = \omega \left|{\theta}\right.} \right)$ specified in (\ref{Eq:Prob_of_symbol_MCMC}), for every $\omega \in \{-1,+1\}$ \par
        denote the new vector by $\hat{\xb}^{(i)}$ \par
        if $\|\yb-\Hb\hat{\xb}^{(i)}\|_2^2 <\|\yb-\Hb\hat{\xb}\|_2^2$, update $\hat{\xb}:=\hat{\xb}^{(i)}$
    }
    \label{alg:reversibleMCMC}
\end{algorithm}

For this type of MCMC detector, what one cares about is the probability that such an algorithm encounters the true transmitted signal within a certain number of iterations. In general, determining this probability within a certain number of iterations is difficult. However, things are relatively easy when we assume that the underlying Markov chain has mixed to the steady state distribution, which is easy to write down, and therefore in steady state it is easy to determine this probability $P_{en}$ of encountering the true transmitted signal. Therefore, an upper bound on the expected time to find the optimal solution is determined by the mixing time (the time it takes to to reach the steady state) of the underlying Markov chain, and the inverse of the probability $P_{en}$ of encountering the true transmitted signal in the steady state.

We remark that $\alpha$ represents a tunable positive parameter which controls the mixing time of the Markov chain, and this parameter is also sometimes called the ``temperature". If we let $\alpha \rightarrow \infty$, the MCMC detector is a just a uniform random walk in the signal space, namely in each iteration the detector choose constellation points with equal probabilities,  and the underlying Markov chain quickly mixes to its steady state \cite{Levin}. When $\alpha$ is close to $0$, the MCMC detector will eventually ``reside'' at the optimal solution, but it may take a very long time to get there from an initial suboptimal signal vector.

On the one hand, the smaller $\alpha$ is, the larger the stationary probability for the optimal solution will be, and the easier it is for the MCMC detector to find the optimal solution in the stationary distribution.
On the other hand, as $\alpha$ gets smaller, it often takes a long time for the Markov chain to converge to its stationary distribution. In fact, as we will show in the paper, there is often a lower bound on $\alpha$, in order to ensure the fast mixing of the Markov chain to its stationary distribution.

\subsection{Comparisons with conventional MCMC detectors}
Our proposed MCMC detector is different from conventional MCMC detectors \cite{Wang_Poor_03, Zhu_Farhang_Boroujeny_05}.  In \cite{Wang_Poor_03, Zhu_Farhang_Boroujeny_05}, the conditional transition probabilities of the underlying Markov chains were directly determined by the posterior likelihood of data sequences. In other words, the ``temperature'' $\alpha$ of these MCMC detectors is directly set as the standard deviation of channel noise \cite{Wang_Poor_03, Zhu_Farhang_Boroujeny_05}. In this paper, however, we have the freedom of optimizing this temperature parameter $\alpha$.

Our proposed method is also very different from simulated annealing techniques where the temperature is slowly reduced until the detector converges to an acceptable solution. In our MCMC detector, the temperature is set as a \emph{fixed} value, and we care about a fast mixing of the underlying Markov chain to a stationary probability distribution and a big enough probability of encountering the transmitted signal in steady state.

\subsection{Mixing time}
It is not hard to see that the Markov chain of MCMC detector is reversible and has $2^{N}$ states with the stationary distribution $e^{-\frac{1}{2\alpha^2} \left\| \yb - \sqrt{\frac{\SNR}{N}} \Hb \hat{\xb} \right\|^2   }$ (without normalization) for a state $\hat{\xb}$. The $2^N \times 2^N$ transition matrix is denoted by $P$, and the element $P_{i,j}$ in the $i$-th ( $1\leq i \leq N$) row and $j$-th ( $1\leq j \leq N$) column is the probability of transferring to state $j$ conditioned on the previous state is $i$. So each row of $P$ sums up to $1$ and the transition matrix after $t$ iterations is $P^{t}$. We denote the vector for the stationary distribution as $\mathbf{\pi}$. Then for an $\epsilon>0$, the mixing time $t(\epsilon)$ is a parameter describing how long it takes for the Markov chain to get close to the stationary distribution \cite{Levin}, namely,
\begin{equation*}
t_{mix}(\epsilon):=\min\{t: \max_{\tilde{\xb}} \|P^{t}(\tilde{\xb},\cdot)-\mathbf{\pi}\|_{TV}\leq \epsilon\},
\end{equation*}
where $\|\mu-\nu\|_{TV}$ is the usual total variation distance between two distributions $\mu$ and $\nu$ over the state space $\{+1,-1\}^{N}$.
\begin{equation*}
\|\mu-\nu\|_{TV}=\frac{1}{2}\sum_{\zb \in \{+1,-1\}^{N}} |\mu(\zb)-\nu(\zb)|.
\end{equation*}

The mixing time is closely related to the spectrum of the transition matrix $P$. More precisely, for a reversible Markov chain, its mixing time is generally small when the gap between the largest and the second largest eigenvalue of $P$, namely $1-\lambda_2$, is large. The inverse of this gap, $\frac{1}{1-\lambda_2}$, is called the relaxation time for this Markov chain.

\subsection{Sequential Markov Chain Monte Carlo Detectors}
In this paper, for simplicity of implementations, we also consider a sequential MCMC detector, especially in numerical simulations. The only difference between sequential MCMC detectors and reversible MCMC detectors is the way they choose the position index to update. Sequential MCMC detectors can have many block iterations. We define one \emph{block iteration} of the sequential MCMC detector as an sequential update of all the $N$  indices $\left\{1, \ldots, N \right\}$ in the estimated symbol vector $\hat{{\xb}}$, starting from $j=1$ to $j=N$. Namely, in one block iteration, we update $N$ indices. For each index $j$, the updating rule for the sequential MCMC detector is the same as the reversible MCMC detector. We remark, however, that the mixing time results in this paper are only for reversible MCMC detectors.
\subsection{Complexity of the MCMC detector}
\label{subsec:MCMC_complexity}
The conditional probability for the $j$-th symbol in \eqref{Eq:Prob_of_symbol_MCMC} can be computed efficiently by reusing the result obtained in earlier iterations, when we evaluate $\left\| \yb - \sqrt{\SNR/N} \Hb \hat{{\xb}}_{j \left|{\omega}\right.} \right\|^2$. Since we are only changing the $j$-th symbol in the symbol vector, the difference $\db^{(l)} \triangleq \yb - \sqrt{\SNR/N} \Hb \hat{{\xb}}_{j \left|{\omega }\right.} $ can be expressed as
\begin{equation}
\label{EQ:difference_for_norm}
 \db^{(l)} = \db^{(l-1)} - \sqrt{\frac{\SNR}{N}} \hb_{j} \Delta s_{j\left|{\omega}\right.} \ ,
\end{equation}
where $l$ is the index for the number of iterations, $\Delta x_{j\left|{\omega}\right.} \triangleq x_{j \left|{\omega}\right.}^{(l)} - x_{j \left|{ \tilde{\omega} }\right.}^{(l-1)}$, and $\hb_j$ is the $j$-th column of $\Hb$.  Thus, the computation of conditional probability when changing the symbol in the $j$-th position costs $2N$ operations\footnote{We need to compute both the product $\hb_{j} \Delta x_{j\left|{\omega}\right.}$ and the inner product $(\db^{(l)})^T\db^{(l)}$ . }, where we define an operation as a Multiply and Accumulate (MAC) instruction. This leads to a complexity of $O \left( 2N[\left|\Xi \right| -1) ]\right)$ operations per iteration.

\subsection{MCMC sampling using QR- or QL-factorization}
\label{subsec:QR_MCMC}
In the case where the number of iterations in the MCMC detector is sufficiently larger than the system size, the complexity of MCMC detector can be reduced even further using a QR- or QL-factorization of the channel matrix, $\Hb = {\bf \tilde{Q}R} = {\bf QL}$, such that the optimization problem in \eqref{EQ:Original_minimization_problem} becomes
\begin{equation}
\label{EQ:QR_minimization_problem}
 	\mathop {\text{min} }\limits_{s \in \Xi^{N}} \ \  \left\| \tilde{\yb} - \sqrt{\frac{\SNR}{N}} {\bf L} {\xb} \right\|^2 \ ,
\end{equation}
where $\hat{\yb} \triangleq {\bf Q}^T \yb$. Since ${\bf L}$ is a lower triangular matrix, the product ${\bf L}{\xb}$ requires less operations compared to a full channel matrix. Suppose that we need to update position index $j$ at the current iteration and assume $\db^{(l-1)}$ is known, we only need to compute the indices from $j$ to $N$ in $\db^{(l)}$, since these are the only non-zero elements in ${\bf L}_{1:N,j} \Delta x_{j\left|{\Xi}\right.}$. Thus, for a square channel matrix of size $N$ the complexity of one iteration in the MCMC detector can roughly be reduced to half the number of operations, namely $O \left( N[\left|\Xi \right| -1]\right)$. This computation saving should be compared with the complexity of performing the QL-factorization, which requires $O \left( N^3 \right)$, and therefore, in cases where the number of iterations is $k > N^2/\left(\left|\Xi \right| -1\right)$, we can achieve a complexity reduction.

\subsection{Norm-2 MCMC Sampler}
 \label{subsection:norm-2}
In this paper, we also propose a new MCMC detector called norm-2 MCMC detector. For this new MCMC detector, we can more easily lower bound its mixing time. We remark, however, in most parts of this paper, ``the MCMC detector'' refers to the squared-norm-2 MCMC detector in Subsection \ref{subsection:reversible}.

 The norm-2 MCMC detector is mostly identical to the squared-norm-2 MCMC detector, except for the computation of transition probability in \eqref{Eq:Prob_of_symbol_MCMC}. Instead of using  \eqref{Eq:Prob_of_symbol_MCMC}, the transition probability for norm-2 MCMC detector is given by

 \begin{equation}
\label{Eq:Prob_of_symbol_MCMC_norm2}
p\left( {\hat{\xb}_j^{(l+1)} = \omega \left|{ \theta }\right.} \right) =  \frac{e^{-\frac{1}{2\alpha^2} \left\| \yb - \sqrt{\frac{\SNR}{N}} \Hb \hat{\xb}_{j \left|{\omega}\right.} \right\|_2   }}{ \sum\limits_{\hat{\xb}_{j \left|{\tilde{\omega}}\right.} \in \Xi}{e^{-\frac{1}{2\alpha^2} \left\| \yb - \sqrt{\frac{\SNR}{N}} \Hb \hat{\xb}_{j \left|{ \tilde{\omega} }\right.} \right\|_2 } }}  \ ,
\end{equation}
where $\tilde{\xb}_{j \left|{\omega}\right. }^T \triangleq \left[\hat{\xb}_{1:j-1}^{(l)}, \omega, \hat{\xb}_{j+1:N}^{(l)} \right]^T$ and $\theta = \left\{ \hat{\xb}^{(l)}, j, \yb, \Hb \right\}$.  So we only use the $\ell_2$ norm in the exponent of the transition probability.

\section{Probability of Error}
\label{sec:P_error}
First, we would like to derive the probability of error for ML detection in MIMO systems, and then use the results to characterize the SNR regime of interest.  The error probability is calculated by averaging over the random matrices $H$ and random noises. Before we derive the probability of error for the ML detector, we will state a lemma which we will make repeated use of.
\begin{lemma}[Gaussian Integral] \label{Lem:Gaussian_integral} Let $\vb$ and $\xb$ be independent
  Gaussian random vectors with distribution ${\cal N}(\zerob,\Ib_N)$
  each. Then
\be
\Ec \left\{ e^{\eta\left(\|\vb+a\xb\|^2-\|\vb\|^2\right)} \right\}=
  \left(\frac{1}{1-2a^2\eta(1+2\eta)}\right)^{N/2}.
\ee
\label{lem:g-i}
\end{lemma}
\begin{proof}
\label{Pro:Gaussian_integral}
See Appendix \ref{subsec:Gaussian_integral} for a detailed proof.
\end{proof}
\vspace{2mm}
Let us first look at the probability of error using maximum likelihood
detection. We will make an error if there exists a vector ${\xb}\neq -\oneb$
such that
\[
\left\|\yb-\sqrt{\frac{\SNR}{N}}\Hb{\xb}\right\|^2\leq
\left\|\yb+\sqrt{\frac{\SNR}{N}} \Hb\oneb\right\|^2 = \|\upsib\|^2 \ .
\] In other words,
\begin{eqnarray*}
P_e & = & \mbox{Prob}\left(\left\|\yb-\sqrt{\frac{\SNR}{N}}\Hb{\xb}\right\|^2\leq
  \|\upsib\|^2 \right) \\
& = & \mbox{Prob}\left(\left\|\upsib+\sqrt{\frac{\SNR}{N}}\Hb(-\oneb-{\xb})\right\|^2\leq
  \|\upsib\|^2 \right) \ ,\\
\end{eqnarray*}
for some ${\xb}\neq -\oneb$, which can be formulated as
\begin{eqnarray*}
P_e & = &  \mbox{Prob}\left(\left\|\upsib+2\sqrt{\frac{\SNR}{N}}\Hb\delb\right\|^2\leq
  \|\upsib\|^2 \right) \ ,
\end{eqnarray*}
for some $\delb\neq 0$, where $\delb  \triangleq \frac{1}{2}(-\mathbf{1}-\xb)$. Note that in the above equation $\delb$ is a vector of zeros and
$-1$'s. Now using the union bound
\be
P_e \leq \sum_{\delb\neq 0}
\mbox{Prob}\left(\left\|\upsib+2\sqrt{\frac{\SNR}{N}}\Hb\delb\right\|^2\leq
  \|\upsib\|^2\right).
\ee
We will use the Chernoff bound to bound the quantity inside the
summation. Thus,
\begin{subequations}
\begin{align}
& \mbox{Prob}\left(\left\|\upsib+2\sqrt{\frac{\SNR}{N}}\Hb\delb\right\|^2 \leq \|\upsib\|^2\right) \\
& \leq \Ec \left\{ e^{-\beta  \left(\left\|\upsib+2\sqrt{\frac{\SNR}{N}}\Hb\delb\right\|^2-\|\upsib\|^2\right)} \right\} \\
& = \left(\frac{1}{1+8\frac{\SNR\|\delb\|^2}{N} \beta(1-2\beta)}\right)^{N/2},
\end{align}
\end{subequations}
where $\beta\geq 0$ is the Chernoff parameter, and where we have used
Lemma \ref{lem:g-i} with $\eta = -\beta$ and
$a = 2\sqrt{\frac{\SNR\|\delb\|^2}{N}}$, since
\[ \Ec \left\{ \left(2\sqrt{\frac{\SNR}{N}}\Hb\delb\right)
\left(2\sqrt{\frac{\SNR}{N}}\Hb\delb\right)^T \right\} =
4\frac{\SNR\|\delb\|^2}{N} \Ib_N. \]
The optimal value for $\beta$ is $\frac{1}{4}$, which yields the
tightest bound
\be
\mbox{Prob}\left(\left\|\upsib+2\sqrt{\frac{\SNR}{N}}\Hb\delb\right\|^2\leq
  \|\upsib\|^2\right) \leq \left(\frac{1}{1+\frac{\SNR\|\delb\|^2}{N}}
\right)^{N/2}.
\ee
Note that this depends only on $\|\delb\|^2$, the number of nonzero
entries in $\delb$. Plugging this into the union bound yields
\be
P_e \leq \sum_{i=1}^N \left( \begin{array}{c} N \\ i \end{array}
\right) \left(\frac{1}{1+\frac{\SNR i}{N}}
\right)^{N/2}.
\ee

Let us first look at the linear (i.e., $i$ proportional to $N$) terms
in the above sum. Thus,
\[ \left( \begin{array}{c} N \\ i \end{array}
\right) \left(\frac{1}{1+\frac{\SNR i}{N}}
\right)^{N/2} \approx
e^{N H(\frac{i}{N})-\frac{N}{2}\ln\left(1+\frac{\SNR i}{N}\right)}, \]
where $H(\cdot)$ is entropy in ``nats''. Clearly, if
\[ \lim_{N\rightarrow\infty}\SNR = \infty , \]
then the linear terms go to zero (superexponentially fast).

Let us now look at the sublinear terms. In particular, let us look at
$i=1$:
\[ N\left(\frac{1}{1+\frac{\SNR}{N}}\right)^{N/2} \approx
N e^{-\SNR/2}. \]
Clearly, to have this term go to zero, we require that $\SNR>
2\ln N$.

A similar argument shows that all other sublinear terms also go to
zero, and so we have:

\begin{lemma}[SNR scaling] If $\SNR>2\ln N+f(N)$, where $f(N)$ is an arbitrary function that goes to $\infty$ as $N \rightarrow \infty$, then
  $P_e\rightarrow 0$ as $N\rightarrow\infty$.
\end{lemma}

\section{Computing the optimal $\alpha$}
\label{sec:comp_opt_alpha}
In this section, we derive the optimal value of the ``temperature" parameter which controls the mixing time of the underlying Markov chain. The temperature has the desirable property that once the Markov chain has mixed to steady state, there is only polynomially (and not exponentially) small probability of encountering the optimal solution.
\subsection{Mean of $\pi_{-\oneb}$}
In the following section we compute the expected value of the stationary probabilities of the states, where the expectation is taken over random Gaussian $H$ and noises. More specifically, we are examining the probability of state ${\xb} = -\oneb$, denoted by $\pi_{-\oneb}$ (recall that we assumed that $-\oneb$ is transmitted symbol vector).

This calculation has a lot in common with the one given in Section \ref{sec:P_error}. Note that
\begin{subequations}
\label{Eq:mean_of_pi}
\begin{align}
\pi_{-\oneb} & = \frac{e^{-\frac{1}{2\alpha^2}\left\|\yb+\sqrt{\frac{\SNR}{N}}\Hb \oneb\right\|^2}}
{\sum_{\xb} e^{-\frac{1}{2\alpha^2}\left\|\yb+\sqrt{\frac{\SNR}{N}}\Hb {\xb}\right\|^2}} \\
& = \frac{e^{-\frac{1}{2\alpha^2}\left\|\upsib\right\|^2}}
{\sum_{\xb} e^{-\frac{1}{2\alpha^2}\left\|\upsib+\sqrt{\frac{\SNR}{N}}\Hb({\xb}-\oneb)\right\|^2}} \\
& = \frac{e^{-\frac{1}{2\alpha^2}\left\|\upsib\right\|^2}}
{\sum_{\delb} e^{-\frac{1}{2\alpha^2}\left\|\upsib+2\sqrt{\frac{\SNR}{N}}\Hb \delb\right\|^2}} \\
& = \frac{1} {\sum_{\delb} e^{-\frac{1}{2\alpha^2}\left(
\left\|\upsib+2\sqrt{\frac{\SNR}{N}}\Hb\delb\right\|^2-\|\upsib\|^2\right)}} \ ,
\end{align}
\end{subequations}
where $\delb$ is a vector of zeros and ones.

Now, using Jensen's inequality and the convexity of $\frac{1}{t}$ when $t>0$,
\begin{subequations}
\begin{align}
\Ec \left\{ \pi_{-\oneb} \right\} & \geq \frac{1}{\Ec \left\{ \frac{1}{  \pi_{-\oneb} }\right\} } \\
 & =
\frac{1}{\Ec \left\{ \sum_{\delb} e^{-\frac{1}{2\alpha^2}\left(
\left\|\upsib+2\sqrt{\frac{\SNR}{N}}\Hb \delb\right\|^2-\left\|\upsib\right\|^2\right)} \right\}} \\
& = \frac{1}
{\sum_{\delb} \Ec \left\{ e^{-\frac{1}{2\alpha^2}\left(
\left\|\upsib+2\sqrt{\frac{\SNR}{N}}\Hb \delb\right\|^2-\left\|\upsib\right\|^2\right) }\right\} } \\
& = \frac{1}{1+\sum_{\delb \neq 0}
\left(\frac{1}{1+4\frac{\SNR\left\|\delb\right\|^2}{N}
\frac{1}{\alpha^2}(1-\frac{1}{\alpha^2})}\right)^{N/2}} \label{Eq:Jensen_part4} \\
& = \frac{1}{1+\sum_{i=1}^N\left(\begin{array}{c} N \\ i \end{array} \right)
\left(\frac{1}{1+\frac{\beta i}{N}}\right)^{N/2}} \label{Eq:Jensen_part5} \ .
\end{align}
\end{subequations}
In \eqref{Eq:Jensen_part4} we have used Lemma \ref{lem:g-i} and in \eqref{Eq:Jensen_part5} we have defined $\beta \triangleq 4\SNR\frac{1}{\alpha^2}(1-\frac{1}{\alpha^2})$.
 While it is possible to focus on the linear and sublinear terms in the
above summation separately, to give conditions for $\Ec \left\{ \pi_{-\oneb} \right\}$ to have
the form of $1/\mbox{poly}(N)$, we will be interested in the exact
exponent of the poly and so we need a more accurate estimate. To do
this we shall use saddle point integration. To this end, note that
\begin{equation*}
\left(\begin{array}{c} N \\ i \end{array} \right)
\left(\frac{1}{1+\frac{\beta i}{N}}\right)^{N/2} \\
= e^{N H(\frac{i}{N})-\frac{N}{2}\ln\left(1+\frac{\beta i}{N}\right)} \ ,
\end{equation*}
again $H(\cdot)$ represents the entropy in ``nats''. And so the summation in the denominator of \eqref{Eq:Jensen_part5} can be
approximated as a Stieltjes integral:
\begin{subequations}
\begin{align}
\sum_{i=1}^N\left(\begin{array}{c} N \\ i \end{array} \right)
\left(\frac{1}{1+\frac{\beta i}{N}}\right)^{N/2} & = N\sum_{i=1}^Ne^{N H(\frac{i}{N})-\frac{N}{2}\ln\left(1+\frac{\beta i}{N}\right)}\frac{1}{N} \\
& = N\int_0^1e^{N H(x)-\frac{N}{2}\ln\left(1+\beta x\right)}dx \ .
\end{align}
\label{stieltjesintegral}
\end{subequations}
For large $N$, this is a saddle point integral and can be approximated
by the formula
\be
\int_0^1e^{N f(x)}dx \approx \sqrt{\frac{2\pi}{N|f''(x_0)|}}e^{N f(x_0)} \ ,
\label{Eq:saddlepoint}
\ee
where $x_0$ is the saddle point of $f(\cdot)$, i.e.,$f'(x_0) = 0.$
In our case,
\[ f(x) = -x\ln x-(1-x)\ln(1-x)-\frac{1}{2}\ln(1+\beta x) \ , \]
and so
\[ f'(x) = \ln\frac{1-x}{x}-\frac{1}{2}\frac{\beta}{1+\beta x} \ . \]
In general, it is not possible to solve for $f'(x_0) = 0$ in closed
form. However, in our case, we assume that $\beta = 4\SNR
\frac{1}{\alpha^2}(1-\frac{1}{\alpha^2})\gg 1$ (In fact, we must have $\beta\rightarrow \infty$ as $N\rightarrow \infty$. Otherwise, (\ref{stieltjesintegral}) will be exponential in $N$ ). In this case, it is not too hard to verify that the saddle point is given by
\be
x_0 \approx e^{-\frac{\beta}{2}} \ .
\ee
And hence
\begin{equation*}
\begin{split}
f(x_0) & = -e^{-\frac{\beta}{2}}\ln e^{-\frac{\beta}{2}} -(1-e^{-\frac{\beta}{2}})\ln(1-e^{-\frac{\beta}{2}}) \\
& \quad - \frac{1}{2}\ln(1+\beta e^{-\frac{\beta}{2}}) \\
& \approx
\frac{\beta}{2}e^{-\frac{\beta}{2}}+e^{-\frac{\beta}{2}}-\frac{1}{2} \beta e^{-\frac{\beta}{2}} \\
& = e^{-\frac{\beta}{2}} \ ,
\end{split}
\end{equation*}
and further plugging $x_0$ into
\[ f''(x) = -\frac{1}{x} -\frac{1}{1-x}-\frac{1}{2}\frac{\beta^2}{(1+\beta x)^2} \ , \]
yields
\be
f''(x_0) \approx -e^{\frac{\beta}{2}}-1+\frac{1}{2}\beta^2 \approx
-e^{\frac{\beta}{2}} \ .
\ee
Replacing these into the saddle point expression in \eqref{Eq:saddlepoint} show that
\begin{equation}
\label{Eq:sum_approx_by_saddl}
\sum_{i=1}^N \left(\begin{array}{c} N \\ i \end{array} \right)
\left(\frac{1}{1+\frac{\beta i}{N}}\right)^{N/2} \approx
\sqrt{2\pi/N} \exp{\left(N e^{-\frac{\beta}{2}}-\frac{\beta}{4}\right)} \ .
\end{equation}
We want $\Ec \left\{ \pi_{-\oneb} \right\}$ to behave as $\frac{1}{N^\zeta}$ and according to \eqref{Eq:mean_of_pi} this means
that we want the expression in \eqref{Eq:sum_approx_by_saddl} to behave as $N^{\zeta}$, where $\zeta$ is a positive number. Let us take
\[ e^{Ne^{-\frac{\beta}{2}}} = N^{\zeta} \ . \]
Solving for $\beta$ yields
\be
\beta = 2(\ln N -\ln\ln N -\ln\zeta).
\ee
Incidentally, this choice of $\beta$ yields
\be
e^{-\frac{\beta}{4}} \approx \frac{1}{\sqrt{N}}.
\ee
Finally, this choice of $\beta$ means that we have
\[ 4\SNR \frac{1}{\alpha^2}\left(1-\frac{1}{\alpha^2}\right) =
2(\ln N -\ln\ln N -\ln\zeta) \ , \]
and so we have the following result.

\begin{lemma}[Mean of $\pi_{-\oneb}$] As $N\rightarrow \infty$, if $\alpha$ is chosen such that
\be
\label{EQ:Alpha_expression}
\frac{\alpha^2}{1-\frac{1}{\alpha^2}} = \frac{4\SNR}
{2(\ln N -\ln\ln N -\ln\zeta)} \ ,
\ee
then
\be
\label{Eq:Mean_invese_pi_lemma}
\Ec \left\{\frac{1}{\pi_{-\oneb}} \right\} \leq N^{\zeta} \ .
\ee
and
\be
\label{Eq:Mean_pi_lemma}
\Ec \left\{\pi_{-\oneb} \right\} \geq N^{-\zeta} \ .
\ee
\end{lemma}

When we have an upper bound on $\Ec \left\{\frac{1}{\pi_{-\oneb}} \right\}$, we can then use the Markov inequality to give upper bounds on the probability that $\frac{1}{\pi_{-\oneb}}$ exceeds a certain threshold. More precisely, we have
$P( \frac{1}{\pi_{-\oneb}} > N^{\gamma'})\leq \Ec \left\{\frac{1}{\pi_{-\oneb}} \right\}/N^{\gamma'} \leq  N^{-(\gamma' - \zeta)}$
for any $\gamma'$. This means that with probability close to $1$ as $N \rightarrow \infty$, the expected time to encounter the transmitted signal in steady state is no bigger than $N^{\gamma'}$, for every $\gamma'>\zeta$.

\subsection{Value of $\alpha$}
In this subsection we investigate how $\alpha$ behaves as a function of the $\SNR$ and the system dimension, if $\alpha$ is chosen
according to \eqref{EQ:Alpha_expression}. In general, the larger
$\alpha$ is, the faster the Markov chain mixes. However, choosing $\alpha$ any larger than this means that the probability of finding
the optimal solution in stationary distribution is exponentially small. Thus, when choosing the value of $\alpha$, there is a trade-off between faster mixing time of the Markov chain (due to an increase of $\alpha$), and faster encountering the optimal solution in stationary distribution.
In the following, we evaluate \eqref{EQ:Alpha_expression} with $\zeta = i$, denoted as $\alpha_{\zeta=i}$ and we also approximate $\alpha$ in \eqref{EQ:Alpha_expression} by neglecting the terms $\ln\ln(N)$ and $\ln(\zeta)$, leading to
\begin{equation}
	\label{EQ:Alpha_approx}
	\frac{\tilde{\alpha}^4}{\tilde{\alpha}^2-1} = \frac{2 \SNR}{\ln(N)} \ .	
\end{equation}
From \eqref{EQ:Alpha_approx} we see that
\begin{equation}
	\label{EQ:Alpha_approx_quad_equation}
	\alpha^2 =  \frac{\SNR}{\ln(N)} \pm \sqrt{\left(\frac{\SNR}{\ln(N)}\right)^2 - 2\frac{\SNR}{\ln(N)}} \ ,
\end{equation}
which implies that $\tilde{\alpha}$ becomes complex when $\SNR < 2\ln(N)$. However, as stated in Section \ref{sec:System_model} we focus on $\SNR > 2 \ln(N)$. Since we are solving a quadratic equation we get two values of $\alpha^2$, representing the region in which \eqref{Eq:Mean_pi_lemma} is satisfied. Based on the considerations given above, we prefer the value of $\alpha^2$ obtained by the plus sign in \eqref{EQ:Alpha_approx_quad_equation}, denoted $\alpha^2_+$, in order to achieve the fastest mixing time.
In Figure \ref{fig:3_alphas_vs_N_SNR=10dB} the values of $\alpha_{\zeta=2}$, $\alpha_{\zeta=1}$, and $\alpha_{+}$ have been plotted as a function of the system dimension for $\SNR = 10$ dB. Our simulations also suggest that the computed value of $\alpha$ is very close to the optimal choice, even in the case where the condition $\SNR>2\ln(N)$ is not satisfied.
\begin{figure}[tbp]
 \center
\includegraphics[width=2.8in]{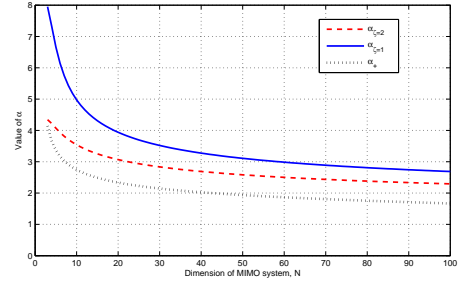}
\caption{Value of $\alpha$ vs. system size $N$, for $\SNR$ = 10 dB.}
\label{fig:3_alphas_vs_N_SNR=10dB}
\end{figure}
\subsection{Mixing time of Markov Chain}
So far, we have examined the largest possible $\alpha$ such that the optimal sequence has a reasonable stationary probability. However, all this was based on assuming that we have reached the stationary distribution.  As $\alpha$ also affects the speed of getting to the stationary distribution, it is interesting to quantify the mixing time of the Markov chain of MCMC detectors.  In the next sections, we will discuss how the mixing time is related to $\alpha$ and the underlying lattice structures in ILS problems.  
\section{Mixing Time with Orthogonal Matrices}
\label{sec:mixing_time_ortho}
 Starting from this section, we consider the mixing time for MCMC for ILS problems and study how the mixing time for ILS problem depends on the linear matrix structure and $\SNR$. As a first step, we consider a linear matrix $\Hb$ with orthogonal columns. As shown later, the mixing time for this matrix has an upper bound independent of $\SNR$. In fact, this is a general phenomenon for ILS problems without local minima.

For simplicity, we use ${\acute{\Hb}}$ to represent $\sqrt{\frac{\SNR}{N}}{\Hb}$ , and the model we are currently considering is
\begin{equation}
\label{eq:orthogonal}
\yb={\acute{\Hb}} \xb +\upsib.
\end{equation}
When the $\SNR$ increases, we simply increase the amplitude of elements in ${\acute{\Hb}}$. We will also incorporate the $\SNR$ term into ${\acute{\Hb}}$ this way in the following sections unless stated otherwise. 

\begin{theorem}
Independent of the temperature $\alpha$ and $SNR$, the mixing time of the MCMC detector for orthogonal-column ILS problems is upper bounded by $N \log(N)+\log(1/\epsilon)N$.
\end{theorem}

This theorem is an extension of the mixing time for regular random walks on an $N$-dimensional hypercube \cite{Levin}. The only difference here is that the transition probability follows (\ref{Eq:Prob_of_symbol_MCMC}) and that the transition probability depends on $\SNR$. Under orthogonal columns, the ILS problem has no local minimum, since $\Hb^{T} \Hb$ is a diagonal matrix in the expansion of $\|\yb-\acute{\Hb}\xb\|_2^2$.

\begin{proof}
When the $j$-th index was selected for updating in the MCMC detector, since the columns of ${\acute{\Hb}}$ are orthogonal to each other, the probability of updating $\xb_{j}$ to $-1$ is $\frac{1}{1+e^{\frac{2\yb^{T} {\bf{h}}_{j}}{\alpha^2}}}$. We note that this probability is independent of the current state of Markov chain $\hat{\xb}$. So we can use the classical coupling idea to get an upper bound on the mixing time of this Markov chain.

Consider two separate Markov chains starting at two different states $\xb_{1}$ and $\xb_{2}$. These two chains follow the same update rule according to (\ref{Eq:Prob_of_symbol_MCMC}). By using the same random source, in each step they select the same position index for updating, and they update that position to the same symbol. Let $\tau_{couple}$ be the first time the two chains come to the same state. Then by a classical result, the total variation distance
\begin{equation}
\label{Eq:TV_coupling}
d(t)=\max_{\tilde{\xb}} \|P^{t}(\tilde{\xb},\cdot)-\mathbf{\pi}\|_{TV} \leq \max_{\xb_{1},\xb_{2}} p_{\xb_1,\xb2} \{\tau_{couple}>t\}.
\end{equation}
Note that the coupling time is just time for collecting all of the positions where $\xb_1$ and $\xb_2$ differ, as in the coupon collector problem. From the famous coupon collector problem \cite{Levin}, for any $\xb_1$ and $\xb_2$,
\begin{equation}
\label{Eq:TV_coupling2}
d(N \log(N)+cN)\leq  p_{\xb_1,\xb_2} \{\tau_{couple}>N \log(N)+cN \} \leq e^{-c}.
\end{equation}
So the conclusion follows.
\end{proof} 
\section{Mixing Time with local Minima}
\label{sec:mixing_time_local}
In this section, we consider the mixing time for ILS problems which have local minima besides the global minimum point. Our main results are that local minima greatly affect the mixing time of MCMC detectors, and rigorous statements are given in Theorem \ref{thm:mixing_scale_alpha}. First, we give the definition of a local minimum.

\begin{definition}
A local minimum $\tilde{\xb}$ is a state such that $\tilde{\xb}$ is not a global minimizer for $\min_{\sb \in \{-1,+1\}^N}\|\yb-{\acute{\Hb}}\sb\|^2$; and any of its neighbors which differ from $\tilde{\xb}$ in only one position index, denoted by $\tilde{\xb}'$, satisfies $\|\yb-{\acute{\Hb}}\tilde{\xb}'\|^2>\|\yb-{\acute{\Hb}}\tilde{\xb}\|^2$.
\end{definition}

We will use the following theorem about the spectral gap of Markov chain to evaluate the mixing time. 
\begin{theorem}[Jerrum and Sinclair (1989) \cite{JerrumSinclair}, Lawler and Sokal (1988) \cite{Lawler}, \cite{Levin}] Let $\lambda_2$ be the second largest eigenvalue of a reversible transition matrix $P$, and let $\gamma=1-\lambda_2$. Then
\begin{equation*}
                \frac{\Phi_{*}^2}{2}\leq \gamma \leq 2\Phi_{*},
\end{equation*}
where $\Phi_{*}$ is the bottleneck ratio (also called conductance, Cheeger constant, and isoperimetric constant) defined as
\begin{equation*}
\Phi_{*}=\min_{\pi(S) \leq \frac{1}{2}} \frac{Q(S,S^{c})}{\pi(S)}.
\end{equation*}
Here $S$ is any subset of the state spaces with stationary measure no bigger than $\frac{1}{2}$, $S^{c}$ is its complement set, and $Q(S,S^c)$ is the probability of moving from $S$ to $S^c$ in one step when starting with the stationary distribution.
\label{thm:gap_bottle}
\end{theorem}

\begin{theorem}
If there is a local minimum ${\tilde{\xb}}$ in an integer least-squares problem and we denote its neighbor differing only at the $j$-th ($1\leq j \leq N$) location as ${\tilde{\xb}}_j$, then the mixing time of the MCMC detector is at least
\begin{equation}
t_{mix}(\epsilon) \geq \log(\frac{1}{2\epsilon})(\frac{1}{\gamma}-1),
\end{equation}
where
\begin{equation}
\gamma=\sum_{j=1}^{N} \frac{2}{N} {\frac{e^{-\frac{\|\yb-{\acute{\Hb}}{\tilde{\xb}}_{j}\|^2}{2\alpha^2}}}{e^{-\frac{\|\yb-{\acute{\Hb}}{\tilde{\xb}}_{j}\|^2}{2\alpha^2}}+e^{-\frac{\|\yb-{\acute{\Hb}}{\tilde{\xb}}\|^2}{{2\alpha^2}}}} }
\end{equation}

The parameter $\gamma$ is upper bounded by
\begin{equation}
\frac{2}{1+e^{\frac{\min_{j}{\|\yb-{\acute{\Hb}}{\tilde{\xb}}_{j}\|^2}-\|\yb-{\acute{\Hb}}{\tilde{\xb}}\|^2}{2\alpha^2}}}
\end{equation}
\label{thm:gap_local}
\end{theorem}

\begin{proof}
We apply Theorem \ref{thm:gap_bottle} to prove this result. We take a local minimum point ${\tilde{\xb}}$ as the single element in the bottle-neck set $S$. Since ${\tilde{\xb}}$ is a local minimum, $\pi(S) \leq \frac{1}{2}$.
\begin{equation}
Q(S,S^{c})=\frac{\pi(S)}{N}\sum_{j=1}^{N} {\frac{e^{-\frac{\|\yb-{\acute{\Hb}}{\tilde{\xb}}_{j}\|^2}{2\alpha^2}}}{e^{-\frac{\|\yb-{\acute{\Hb}}{\tilde{\xb}}_{j}\|^2}{2\alpha^2}}+e^{-\frac{\|\yb-{\acute{\Hb}}{\tilde{\xb}}\|^2}{{2\alpha^2}}}} }
\end{equation}

Dividing by $\pi(S)$, by the definition of $\Phi_{*}$
\begin{equation}
\Phi_{*}\leq \frac{Q(S,S^{c})}{\pi(S)}=\frac{1}{N}\sum_{j=1}^{N} {\frac{e^{-\frac{\|\yb-{\acute{\Hb}}{\tilde{\xb}}_{j}\|^2}{2\alpha^2}}}{e^{-\frac{\|\yb-{\acute{\Hb}}{\tilde{\xb}}_{j}\|^2}{2\alpha^2}}+e^{-\frac{\|\yb-{\acute{\Hb}}{\tilde{\xb}}\|^2}{{2\alpha^2}}}} }
\end{equation}

So we know $\gamma \leq 2 \frac{1}{N}\sum_{j=1}^{N} {\frac{e^{-\frac{\|\yb-{\acute{\Hb}}{\tilde{\xb}}_{j}\|^2}{2\alpha^2}}}{e^{-\frac{\|\yb-{\acute{\Hb}}{\tilde{\xb}}_{j}\|^2}{2\alpha^2}}+e^{-\frac{\|\yb-{\acute{\Hb}}{\tilde{\xb}}\|^2}{{2\alpha^2}}}} }$.
From a well-known theorem for the relationship between $t_{mix}(\epsilon)$ and $\gamma$:
$t_{mix}(\epsilon) \geq (\frac{1}{\gamma}-1) \log(\frac{1}{2\epsilon})$ \cite{Levin},
our conclusion follows.
\end{proof}

\begin{theorem}
For an integer least-squares problem $\min_{\sb \in \{-1,+1\}^N}\|\yb-{\acute{\Hb}}\sb\|^2$, where $\acute{\Hb}$ is fixed and no two vectors give the same objective distance, the relaxation time (the inverse of the spectral gap) of the Markov chain for the reversible MCMC detector (Algorithm \ref{alg:reversibleMCMC}) is upper bounded by a constant as the temperature $\alpha \rightarrow 0$ if and only if there is no local minimum. Moreover, when there is a local minimum, as $\alpha \rightarrow 0$, the mixing time of Markov chain $t_{mix}(\epsilon) =e^{\Omega(\frac{1}{2\alpha^2})}$.
\label{thm:mixing_scale_alpha}
\end{theorem}

{\bf Remarks}: For the signal model $\yb = \sqrt{\frac{\SNR}{N}}{\Hb{\xb}} + \upsib $, if $\alpha$ is set equal to the noise variance as in \cite{Wang_Poor_03, Zhu_Farhang_Boroujeny_05}, it is equivalent to setting ``$\alpha \rightarrow 0$'' when $\SNR \rightarrow \infty$, since in Theorem \ref{thm:mixing_scale_alpha} the SNR term is incorporated into $\acute{\Hb}$ and we keep $\acute{\Hb}$ fixed in Theorem \ref{thm:mixing_scale_alpha}.

\begin{proof}
First, when there is a local minimum, from Theorem \ref{thm:gap_local} and Theorem \ref{thm:gap_bottle}, the spectral gap $\gamma$ is lower bounded by
\begin{equation}
\gamma=\frac{2}{N}\sum_{j=1}^{N} {\frac{e^{-\frac{\|\yb-{\acute{\Hb}}{\tilde{\xb}}_{j}\|^2}{2\alpha^2}}}{e^{-\frac{\|\yb-{\acute{\Hb}}{\tilde{\xb}}_{j}\|^2}{2\alpha^2}}+e^{-\frac{\|\yb-{\acute{\Hb}}{\tilde{\xb}}\|^2}{{2\alpha^2}}}} }
\end{equation}

As the temperature $\alpha \rightarrow 0$,  the spectral gap upper bound
\begin{equation}
\frac{2}{1+e^{\frac{\min_{j}{\|\yb-{\acute{\Hb}}{\tilde{\xb}}_{j}\|^2}-\|\yb-{\acute{\Hb}}{\tilde{\xb}}\|^2}{2\alpha^2}}}
\end{equation}
decreases at the speed of $\Theta(e^{-\frac{\min_{j}{\|\yb-{\acute{\Hb}}{\tilde{\xb}}_{j}\|^2}-\|\yb-{\acute{\Hb}}{\tilde{\xb}}\|^2}{2\alpha^2}})$. So the relaxation time of the MCMC is lower bounded by $t_{mix}(\epsilon) =e^{\Omega(\frac{1}{2\alpha^2})}$, which grows unbounded as $\alpha \rightarrow 0$.

Suppose instead that there is no local minimum. We argue that as $\alpha \rightarrow 0$, the spectral gap of this MCMC is lower bounded by some constant independent of $\alpha$. Again, we look at the bottle neck ratio and use Theorem \ref{thm:gap_bottle} to bound the spectral gap.

Consider any set $S$ of sequences which do not include the global minimum point ${\xb^*}$. As $\alpha \rightarrow 0$, the measure of this set of sequences $\pi(S)\leq \frac{1}{2}$. Moreover, as $\alpha \rightarrow 0$, any set $S$ with $\pi(S)\leq \frac{1}{2}$ can not contain the global minimum point ${\xb^*}$. Now we look at the sequence ${\tilde{\xb}}'$ which has the smallest distance $\|\yb-{\acute{\Hb}}{\tilde{\xb}}'\|$ among the set $S$. Since there is no local minimum, ${\tilde{\xb}}'$ must have at least one neighbor ${\tilde{\xb}}''$ in $S^{c}$ which has smaller distance than ${\tilde{\xb}}'$. Otherwise, this would imply ${\tilde{\xb}}'$ is a local minimum.  So
\begin{equation}
Q(S,S^{c}) \geq \pi({\tilde{\xb}}') \times \frac{1}{N} {\frac{e^{-\frac{\|\yb-{\acute{\Hb}}{\tilde{\xb}}''\|^2}{2\alpha^2}}}{e^{-\frac{\|\yb-{\acute{\Hb}}{\tilde{\xb}}''\|^2}{2\alpha^2}}+e^{-\frac{\|\yb-{\acute{\Hb}}{\tilde{\xb}}'\|^2}{{2\alpha^2}}}} }
\end{equation}

As $\alpha \rightarrow 0$, $\frac{\pi({\tilde{\xb}}')}{\pi(S)} \rightarrow 1$. So for a given $\epsilon>0$, as $\alpha \rightarrow 0$
\begin{equation}
   \frac{Q(S,S^{c})}{\pi(S)} \geq \frac{1-\epsilon}{N} {\frac{e^{-\frac{\|\yb-{\acute{\Hb}}{\tilde{\xb}}''\|^2}{2\alpha^2}}}{e^{-\frac{\|\yb-{\acute{\Hb}}{\tilde{\xb}}''\|^2}{2\alpha^2}}+e^{-\frac{\|\yb-{\acute{\Hb}}{\tilde{\xb}}'\|^2}{{2\alpha^2}}}} },
\end{equation}
which approaches $\frac{(1-\epsilon)}{N}$ as $\alpha \rightarrow 0$ because $\|\yb-{\acute{\Hb}}{\tilde{\xb}}''\|^2 < \|\yb-{\acute{\Hb}}{\tilde{\xb}}'\|^2$.

From Theorem \ref{thm:gap_bottle}, the spectral gap $\gamma$ is at least $\frac{(\frac{Q(S,S^{c})}{\pi(S)})^2}{2}$, which is lower bounded by a constant as $\alpha \rightarrow 0$.
\end{proof}

So from the analysis above, the mixing time is closely related to whether there are local minima in the problem. In the next section, we will see there often exist local minima, which implies very slow convergence rate for MCMC when the temperature is kept at the noise level in the high SNR regime. 
\section{Presence of Local Minima}
\label{sec:local_minimum}
We have seen that the mixing time of MCMC detectors are closely related to the existence of local minima. It is natural to ask how often local minima occur in ILS problems. In this section, we derive some results about how many local minima there are in an ILS problem, especially when the $\SNR$ is high.

\begin{theorem}
There can be exponentially many local minima in an integer least-quare problem
\label{thm:manylocalminima}
\end{theorem}

\begin{proof}
See Appendix \ref{subsec:proofmanylocalminima} for a detailed proof.
\end{proof}

Now we study how often we encounter a local minimum in specific ILS problem models. Without loss of generality, we assume that the transmitted sequence is an all $-\oneb$ sequence. We first give the condition for ${\tilde{\xb}}$ to be a local minimum. We assume that ${\tilde{\xb}}$ is a vector which has $k$ `$+1$' over an index set $K$ with $|K|=k$ and $(N-k)$ `$-1$' over the set $\overline{K}=\{1,2,...,N\}\setminus {K}$.
\begin{lemma}
${\tilde{\xb}}$ is a local minimum if and only if ${\tilde{\xb}}$ is not a global minimum; and
\begin{itemize}
\item  $\forall i \in K$,
\begin{eqnarray}
{\hb}_{i}^{T}(\sum_{j \in K}{\hb}_{j}-\frac{\upsib}{2})<\frac{\|{\hb}_{i}\|^2}{2}
\end{eqnarray}
\item  $\forall i \in \overline{K}$,
\begin{eqnarray}
{\hb}_{i}^{T}(\sum_{j \in K}{\hb}_{j}-\frac{\upsib}{2})>-\frac{\|{\hb}_{i}\|^2}{2}.
\end{eqnarray}

\end{itemize}

\label{lemma:localcondition}
\end{lemma}

\begin{proof}
For a position $i \in K$, when we flip ${\tilde{\xb}}_{i}$ to $1$,  $\|\yb-{\acute{\Hb}}{\tilde{\xb}}'\|^2$ is increased, namely,
\begin{eqnarray}
&&\|\yb-{\acute{\Hb}}{\tilde{\xb}}\|^2-\|\yb-{\acute{\Hb}}{\tilde{\xb}}_{\sim i}\|^2 \nonumber \\
&=& \|-2\sum_{j \in K}{\hb}_{j}+\upsib\|^2-\|-2\sum_{j \in K, j \neq i}{\hb}_{j}+\upsib\|^2 \nonumber \\
&=& 4\|{\hb}_{i}\|^2+4{\hb}_{i}^{T}(2\sum_{j \in K, j \neq i}{\hb}_{j}-\upsib)  \nonumber\\
&<&0,
\end{eqnarray}
where $\tilde{\xb}_{\sim i}$ is a neighbor of $\tilde{\xb}$ by changing index $i$.
This means
\begin{eqnarray}
{\hb}_{i}^{T}(\sum_{j \in K}{\hb}_{j}-\frac{\upsib}{2})<\frac{\|{\hb}_{i}\|^2}{2}.
\end{eqnarray}

For a position $i \in \overline{K}$, when we flip ${\tilde{\xb}}_{i}$ to $-1$,  $\|\yb-{\hb}{\tilde{\xb}}'\|^2$ is also increased, namely,
\begin{eqnarray}
&&\|\yb-{\acute{\Hb}}{\tilde{\xb}}\|^2-\|\yb-{\acute{\Hb}}{\tilde{\xb}}_{\sim i}\|^2 \nonumber \\
&=& \|-2\sum_{j \in K}{\hb}_{j}+\upsib\|^2-\|-2\sum_{j \in K}{\hb}_{j}-2{\hb}_{i}+\upsib\|^2 \nonumber \\
&=& -4\|{\hb}_{i}\|^2+4{\hb}_{i}^{T}(-2\sum_{j \in K}{\hb}_{j}+\upsib)  \nonumber\\
&<&0.
\end{eqnarray}
This means
\begin{eqnarray}
({\hb}_{i})^{T}(\sum_{j \in K}{\hb}_{j}-\frac{\upsib}{2})>-\frac{\|{\hb}_{i}\|^2}{2}.
\end{eqnarray}
\end{proof}

%
%

It is not hard to see that when $\SNR \rightarrow \infty$, $\upsib$ is comparatively small with high probability, so we have the following lemma.
\begin{lemma}
When $SNR \rightarrow \infty$, ${\tilde{\xb}}$ is a local minimum with high probability, if and only if ${\tilde{\xb}} \neq -\mathbf{1}$; and
\begin{itemize}
\item  $\forall i \in K$,
\begin{eqnarray}
{\hb}_{i}^{T}(\sum_{j \in K}{\hb}_{j})<\frac{\|{\hb}_{i}\|^2}{2}
\end{eqnarray}
\item  $\forall i \in \overline{K}$,
\begin{eqnarray}
{\hb}_{i}^{T}(\sum_{j \in K}{\hb}_{j})>-\frac{\|{\hb}_{i}\|^2}{2}.
\end{eqnarray}

\end{itemize}

\end{lemma}

We now set out to investigate the chance of having a local minimum in MIMO systems. 

\begin{theorem}
Consider a $2 \times 2$ matrix ${\acute{\Hb}}$ whose two columns are uniform randomly sampled from the unit-normed $2$-dimensional vector.
When $\upsib=0$, the probability of there existing a local minimum for such an ${\acute{\Hb}}$ is $\frac{1}{3}$.
\label{thm:22indcolumns}
\end{theorem}

Please see the appendix for its proof. 
%

\begin{theorem}
Consider a $2 \times 2$ matrix ${\acute{\Hb}}$ whose elements are independent $\Nc (0,1)$ Gaussian random variables.
When $\upsib=0$, the probability of there existing a local minimum for such an ${\acute{\Hb}}$ is $\frac{1}{3}-\frac{1}{\sqrt{5}}+\frac{2\arctan(\sqrt{\frac{5}{3}})}{\sqrt{5}\pi}$.
\label{thm:22Gaussian}
\end{theorem}

Please refer to the appendix for its proof.

For higher dimension $N$, it is hard to directly estimate the probability of a vector being a local minimum based on the conditions in Lemma \ref{lemma:localcondition}. Simulation results instead suggest that for large $N$, with high probability, there exists at least one local minimum. We conjecture this is the case, but proof or disproof of it seems nontrivial.

%
%
%
%
%
%
%

\section{Choice of Temperature $\alpha$ in High $\SNR$}
\label{sec:choice_alpha}
In previous sections, we have looked at the mixing time of MCMC for an ILS problem. Now we use the results we have accumulated so far to help choose the appropriate temperature of $\alpha$ to ensure that the MCMC mixes fast and that the optimal solution also comes up fast when the system is in a stationary distribution.

When $\SNR \rightarrow \infty$, the ILS problem will have the same local minima as the case $\upsib=0$. From the derivations and simulations, it is suggested that with high probability there will be at least one local minimum, especially for large problem dimension $N$.

So following from Lemma \ref{thm:mixing_scale_alpha} and the reasoning therein, to ensure there is an upper bound on the mixing time as $\SNR\rightarrow \infty$, the temperature $\alpha$ should at least grow at a rate such that
\begin{equation}
\max_{\tilde{\xb}} \min_{\tilde{\xb}'}\frac{\frac{\SNR}{N} \left(\|-{\acute{\Hb}}\mathbf{1}-\tilde{\xb}'\|^2-\|-\Hb\mathbf{1}-\tilde{\xb}\|^2\right)}{2\alpha^2} \leq C,
\label{eq:criterion}
\end{equation}
where $\tilde{\xb}$ is a local minimum and $\tilde{\xb}'$ is a neighbor of $\tilde{\xb}$, and $C$ is a constant.

This requires that $\alpha^2$ grow as fast as $\Omega(\SNR)$ to ensure fast mixing with the existence of local minima. This explains that if we keep the temperature at the noise level, it will lead to slow convergence in the high SNR regime \cite{Farhang_Boroujeny_06}.

We remark that, for the squared-norm-2 MCMC detector, it is hard to explicitly evaluate $\alpha$ from \eqref{eq:criterion}. However, for the norm-2 MCMC detector, the corresponding criterion for a fast mixing is given by
\begin{equation}
\max_{\tilde{\xb}} \min_{\tilde{\xb}'}\frac{\sqrt{\frac{\SNR}{N}} \left(\|-\Hb\mathbf{1}-\tilde{\xb}'\|_2-\|-\Hb\mathbf{1}-\tilde{\xb}\|_2\right)}{2\alpha^2} \leq C,
\label{eq:criterionfornorm2}
\end{equation}
for some constant $C$.

By the triangular inequality, and the concentration of measure result for Gaussian random variables, with high probability as $N \rightarrow \infty$, for any $\epsilon>0$,
\begin{eqnarray}
&&\max_{\tilde{\xb}} \min_{\tilde{\xb}'}{\sqrt{\frac{\SNR}{N}} \left(\|-\Hb\mathbf{1}-\tilde{\xb}'\|_2-\|-\Hb\mathbf{1}-\tilde{\xb}\|_2\right)} \nonumber\\
&& \leq (1+\epsilon) 2\sqrt{\SNR}.
\label{eq:concentration}
\end{eqnarray}

So for the norm-2 MCMC detector, as long as $\alpha^2 \geq (1+\epsilon)\sqrt{\SNR}/C$, the condition \eqref{eq:criterionfornorm2} will be satisfied with high probability. 

\section{Simulation Results}
\label{sec:sim_results}
In this section we present simulation results for an $N \times N$ MIMO system with a full square channel matrix containing i.i.d. Gaussian entries.
In Figure \ref{fig:Sim_alpha_SNR=10} and Figure \ref{fig:Sim_alpha_SNR=14} the Bit Error Rate (BER) of the sequential MCMC detector has been evaluated as a function of the number of block iterations in a $10 \times 10$ system using a varity of $\alpha$ values. Thereby, we can inspect how the parameter $\alpha$ affects the convergence rate of the MCMC detector and, as a reference, we have included the values of $\alpha$ computed using \eqref{EQ:Alpha_expression} and \eqref{EQ:Alpha_approx}, which can be seen in Table \ref{tab:Selection_alpha}.

\begin{table}[!tb]  
    \centering
    \caption{Theoretical values of $\alpha$ for $N = 10$.}
    \begin{tabular}[!h]{l c c}
        \hline \hline
                \SNR & 10 dB & 14 dB \\
                \hline
         				$\alpha_{+,\zeta=2}$  & 4.98  &  7.99 \\
         				$\alpha_{+,\zeta=1}$	& 3.54  &  5.76 \\
         				$\tilde{\alpha}_{+}$	& 2.74  &  4.56 \\
        \hline
     \end{tabular}
     \label{tab:Selection_alpha}
\end{table}

The performance of the Maximum Likelihood (ML), the Zero-Forcing (ZF), and the Linear Minimum Mean Square Error (LMMSE) detector have also been plotted, to ease the comparison of the MCMC detector with these detectors (Please see \cite{Agrell_et_al_02}, for example, for descriptions of these well-known detectors). It is seen that the MCMC detector outperforms both the ZF and the LMMSE detector after only a few block iterations in all the presented simulations, when the tuning parameter $\alpha$ is chosen properly. Furthermore, it is observed that the parameter $\alpha$ has a huge influence on the convergence rate and that the MCMC detector converges toward the ML solution as a function of the iterations\footnote{It should be noted that the way we decode the symbol vector to a given iteration, is to select the symbol vector with has the lowest cost function in all the iterations up to that point in time.}.
\begin{figure}[tb]
\centering
\includegraphics[width=2.6in]{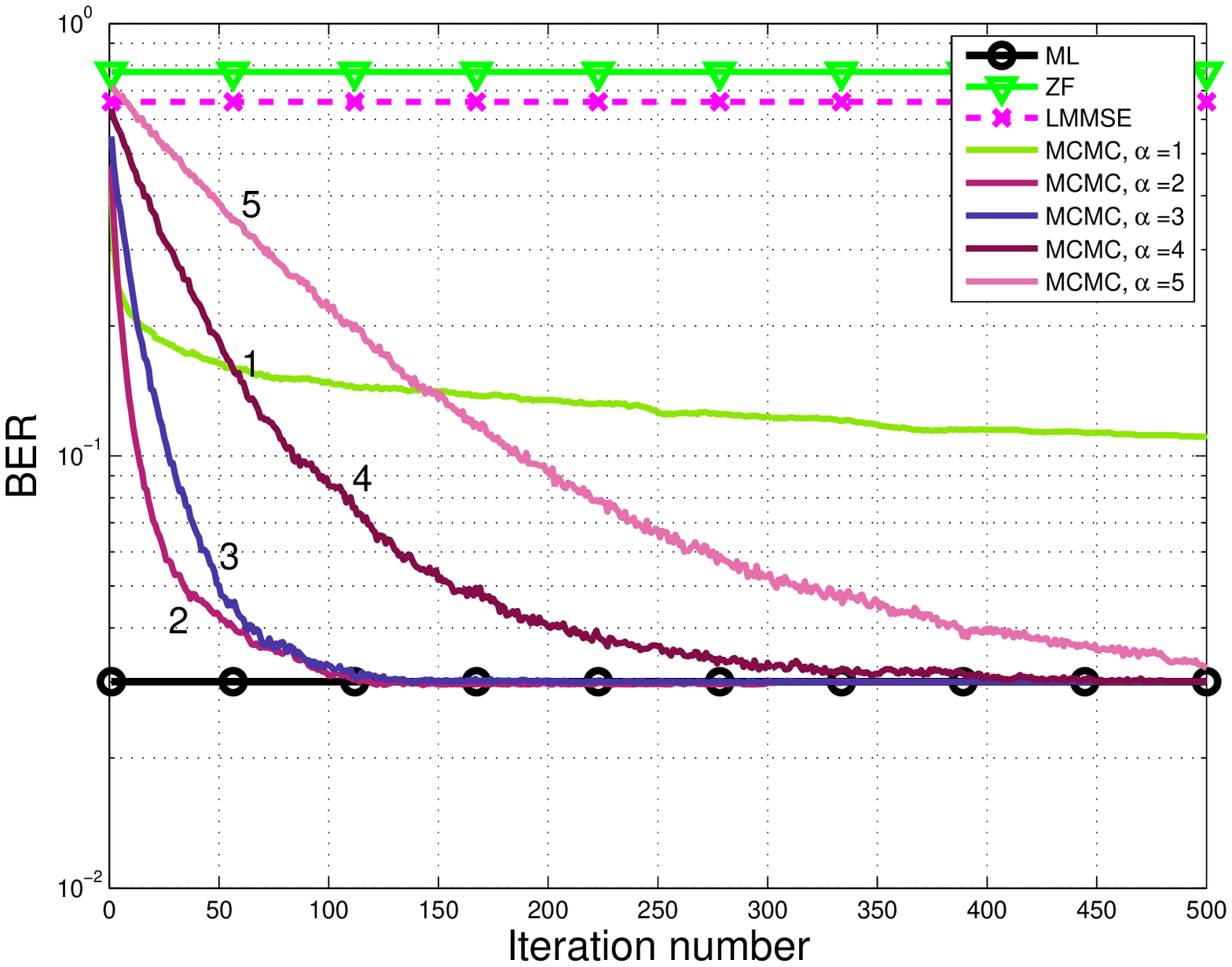}
\caption{BER vs. iterations, $10 \times 10$. $\SNR$ = 10 dB.}
\label{fig:Sim_alpha_SNR=10}
\end{figure}
\begin{figure}[tb]
	\centering
  \includegraphics[width=2.6in]{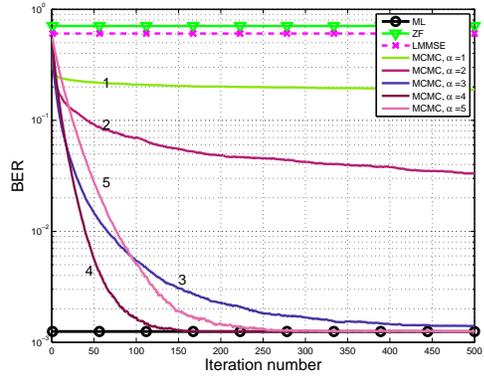}
	\vspace{-2mm}
	\caption{BER vs. iterations, $10 \times 10$ system. $\SNR$ = 14 dB.}
	\label{fig:Sim_alpha_SNR=14}
\end{figure}
Figure \ref{fig:Sim_alpha_BER_vs_SNR_iter=100} shows the BER performance for the MCMC detector for fixed number of iterations, $k = 100$. From the figure we see that the $\SNR$ has a significant influence on the optimal choice of $\alpha$ given a fixed number of iterations. The performance of the sequential MCMC  detector is also shown for a $50 \times 50$ system, which represents a ML decoding problem of huge complexity where an exhaustive search would require $2^{50} \approx 10^{15}$ evaluations. For this problem even the sphere decoder would have an enormous complexity under moderate $\SNR$, and it has actually been proved in \cite{Ottersten_05} that the complexity of SD for $\SNR = O(\ln(N))$ is exponential. Therefore, it has not been possible to simulate the performance of this decoder within a reasonable time and we have therefore ``cheated" a little by initializing the radius of the sphere to the minimum of either the norm of the transmitted symbol vector or the solution found by the MCMC detector. This has been done in order to evaluate the BER performance of the optimal detector. Figure \ref{fig:Sim_alpha_N=50_SNR=12} shows the BER curve as a function of the iteration number while Figure \ref{fig:Sim_alpha_BER_vs_SNR_N=50_iter=500} illustrates the BER curve vs. the $\SNR$. From Figure \ref{fig:Sim_alpha_N=50_SNR=12} we see that there is a quite good correspondence between the simulated $\alpha$ and the theoretical value $\tilde\alpha_+ = 2.6$ obtained from \eqref{EQ:Alpha_approx}.
\begin{figure}[tb]
	\centering
\includegraphics[width=2.6in]{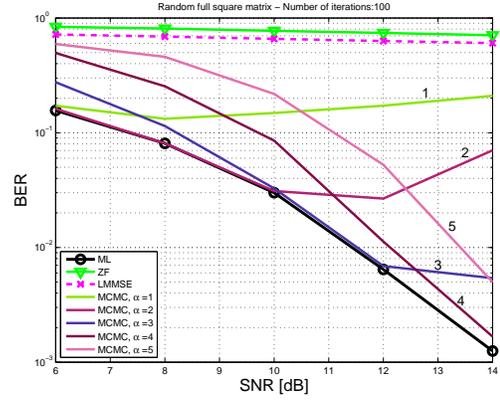}
\vspace{-2mm}
	\caption{BER vs. $\SNR$, $10 \times 10$. Number of iterations, $k = 100$.}
	\label{fig:Sim_alpha_BER_vs_SNR_iter=100}
\end{figure}
\begin{figure}[tb]
	\centering
  \includegraphics[width=2.6in]{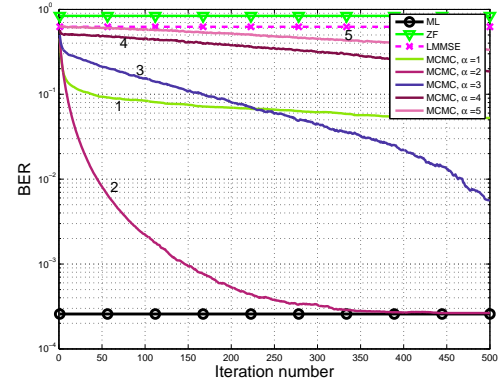}
	\vspace{-2mm}
	\caption{BER vs. iterations, $50 \times 50$ system. $\SNR$ = 12 dB.}
	\label{fig:Sim_alpha_N=50_SNR=12}
\end{figure}
\begin{figure}[tb]
	\centering
\includegraphics[width=2.6in]{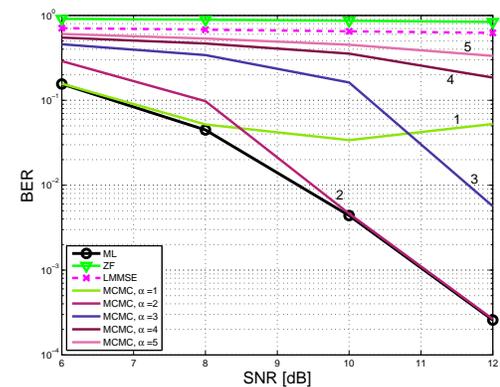}
\vspace{-2mm}
	\caption{BER vs. \SNR, $50 \times 50$ system. Number of iterations, $k = 500$.}
	\label{fig:Sim_alpha_BER_vs_SNR_N=50_iter=500}
\end{figure}

\begin{figure}[!htb]
  \centering
  \includegraphics[width=200pt, height=120pt]{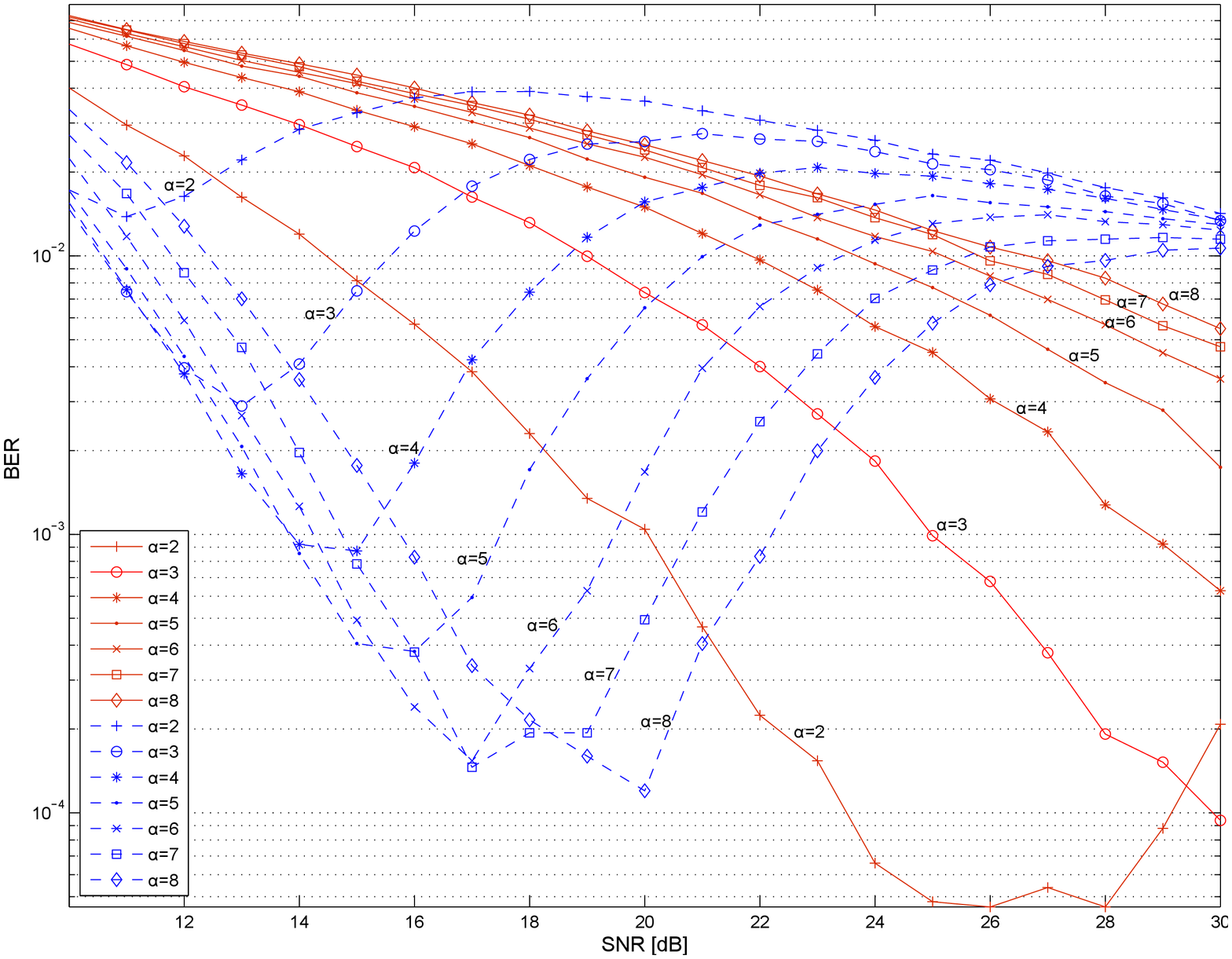}\\
  \caption{BER vs SNR for squared-norm-2 MCMC detector (dashed lines) and norm-2 MCMC detector (solid lines)}\label{1}
\end{figure}

Now we compare the numerical performance of reversible MCMC detectors including square-norm-2 MCMC detector and norm-2 MCMC detector. Again, we simulate $N\times N$ MIMO systems with channel matrices containing zero mean i.i.d Gaussian entries. Figure \ref{1} shows the BER as a function of $\SNR$ for two different MCMC detectors: squared-norm-2 and norm-2 MCMC detector, when $N=10$. 1000 iterations are used in both reversible MCMC detectors which are initialized with a random input vector. Dashed and blue curves in Figure \ref{1} represent squared-norm-2 MCMC detector for various $\alpha$ values. Squared-norm-2 MCMC detector uses $\|\textrm{y}-\sqrt{\frac{\textrm{\SNR}}{\textmd{N}}}\textbf{H}\hat{\sb}\|^{2}$ in the calculation of the probability of transferring from one state to another. Solid and red curves represent norm-2 MCMC detector which uses equation $\|\textrm{y}-\sqrt{\frac{\textrm{\SNR}}{\textmd{N}}}\textbf{H}\hat{\sb}\|$. We can see that for the same $\alpha$ values, norm-2 MCMC detector has better BER compared with the Squared-norm-2 in high SNR.

Now we consider numerical results related to the mixing time of reversible MCMC detectors. In Figure \ref{fig:N_expected}, we plot the expected number of local minima in a system as the problem dimension $N$ grows. For each $N$, we generate $100$ random channel matrices and for each matrix, we examine the number of local minima by exhaustive search. As the problem dimension $N$ grows, the number of local minima grows rapidly.

In Figure \ref{fig:N_frequency}, we plot the probability of there existing a local minimum as the problem dimension $N$ grows. For each $N$, we generated $100$ random channel matrices and for each matrix, we examined whether there exist local minima by exhaustive search. As $N$ grows, the empirical probability of there existing at least one local minimum approaches $1$. It is interesting to see that for $N=2$, our theoretical result $\frac{1}{3}-\frac{1}{\sqrt{5}}+\frac{2\arctan(\sqrt{\frac{5}{3}})}{\sqrt{5}\pi}\approx0.15$ matches well with the simulations.

Figures \ref{2} and \ref{3} show the histograms of the number of local minima for $N=10$ and 12 respectively, under $\SNR=10$. For each parameter $N$, we used exhaustive search to examine the number of local minima in 100 randomly chosen Gaussian channel matrices. Obviously, the average number of local minima increases as $N$ increases, while the frequency of 0 local minima decreases.

%

\begin{figure}[tb]
\centering
\includegraphics[width=3.0in, height=1.75in]{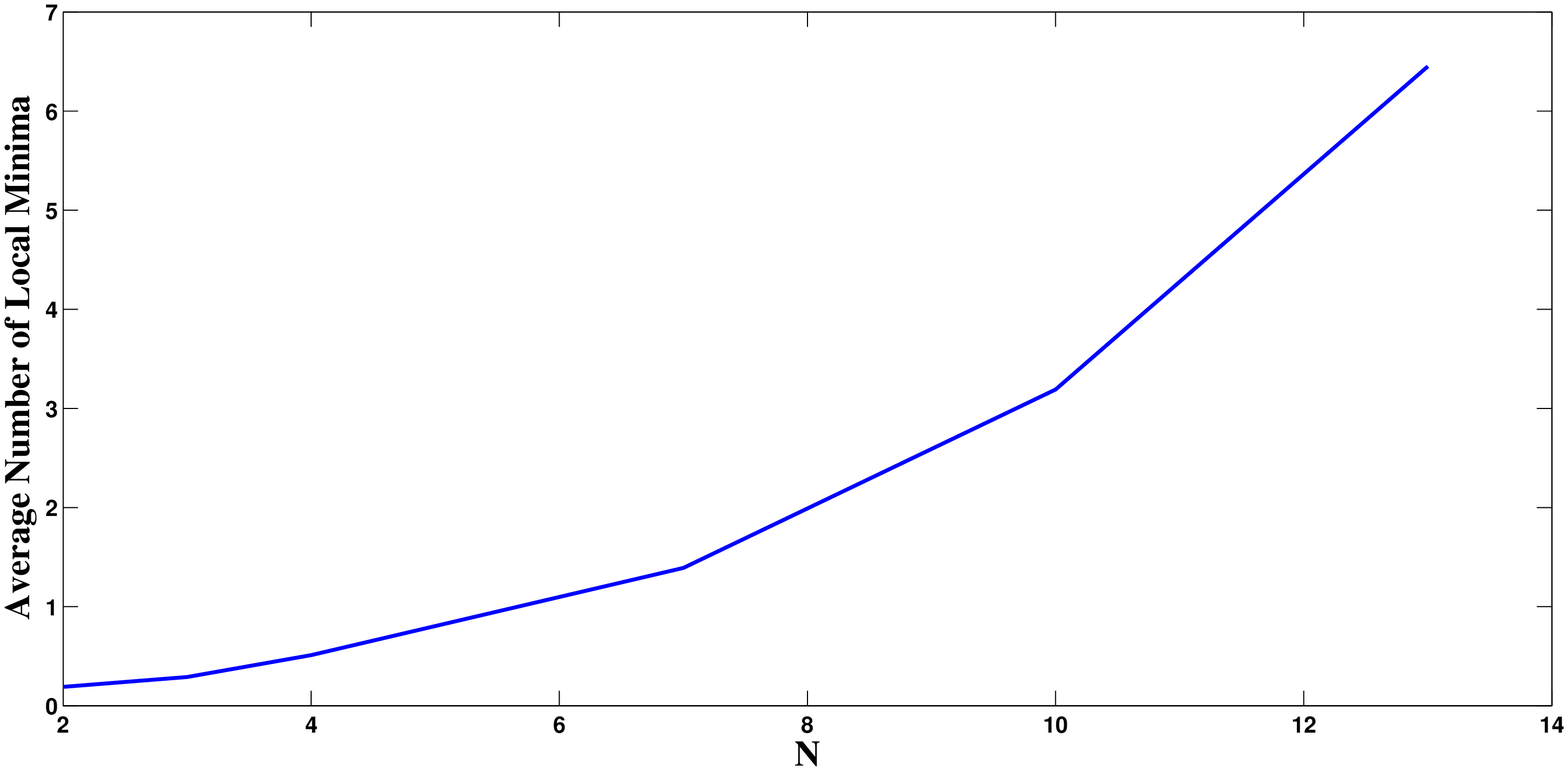}
\caption{Average number of local minima}
\label{fig:N_expected}
\end{figure}
\begin{figure}[tb]
	\centering
  \includegraphics[width=3.0in,height=1.75in]{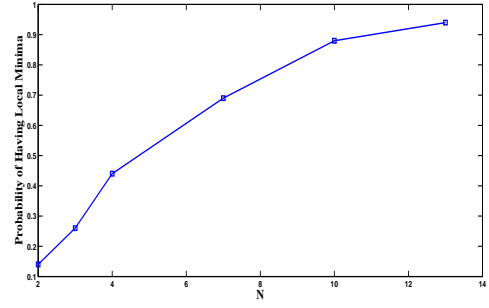}
	\vspace{-2mm}
	\caption{The probability of having local minima}
	\label{fig:N_frequency}
\end{figure}
%
Figure \ref{4} presents the histograms of the spectral gap when there are 0, 1, 2, and 3 local minima respectively for $N=5$ and $\SNR=10$. We generated $10^{5}$ randomly Gaussian channel matrices. In each matrix we examined the number of local minima and calculated the spectral gap when $\alpha^2=1$. For all these figures, each bar represents the percentage of matrices which fall in a spectral gap interval of 0.01. We can see that, when there is 0 local minimum, around 50 percent of the matrices' spectral gap fall between 0.19 and 0.2, suggesting these MCMC detectors mix fast. However, when there is at least one local minimum, a high percentage of the matrices have spectral gap values between 0 and 0.01. This percentage increases with the increasing of the number of local minima.

\begin{figure}[!htb]
  \centering
  \includegraphics[width=150pt, height=100pt]{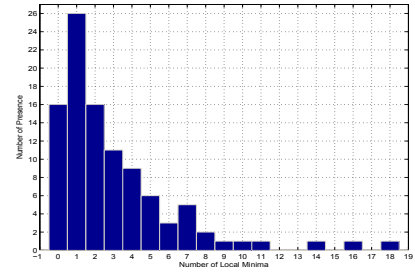}\\
  \caption{Histograms of the number of local minima for N=10}\label{2}
\end{figure}

\begin{figure}[!htb]
  \centering
  \includegraphics[width=150pt, height=100pt]{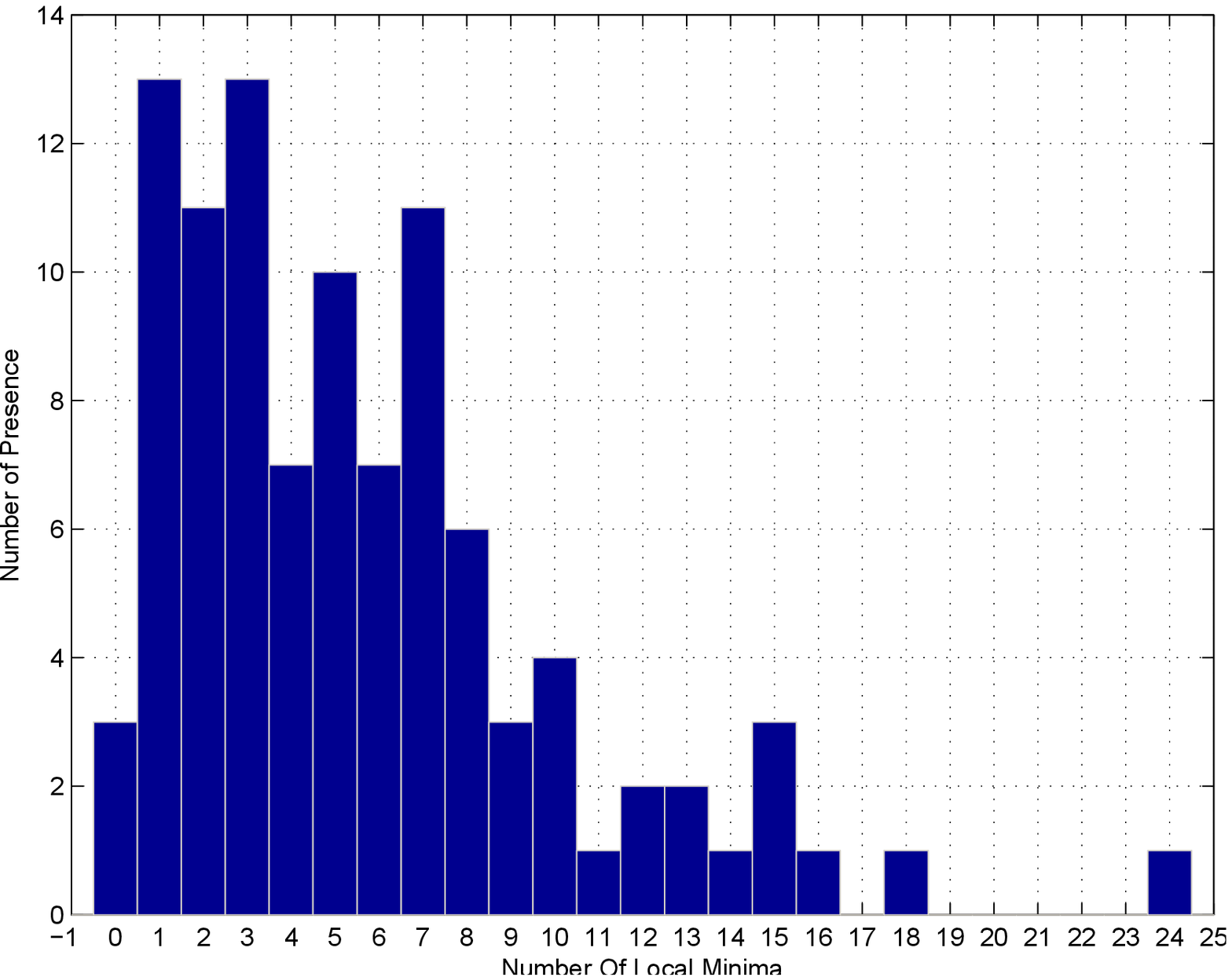}\\
  \caption{Histograms of the number of local minima for N=12}\label{3}
\end{figure}

%
%

\begin{figure}

    \begin{minipage}{.5\linewidth}
        \centering
        \subfloat[]{\label{main:a}\includegraphics[scale=.20]{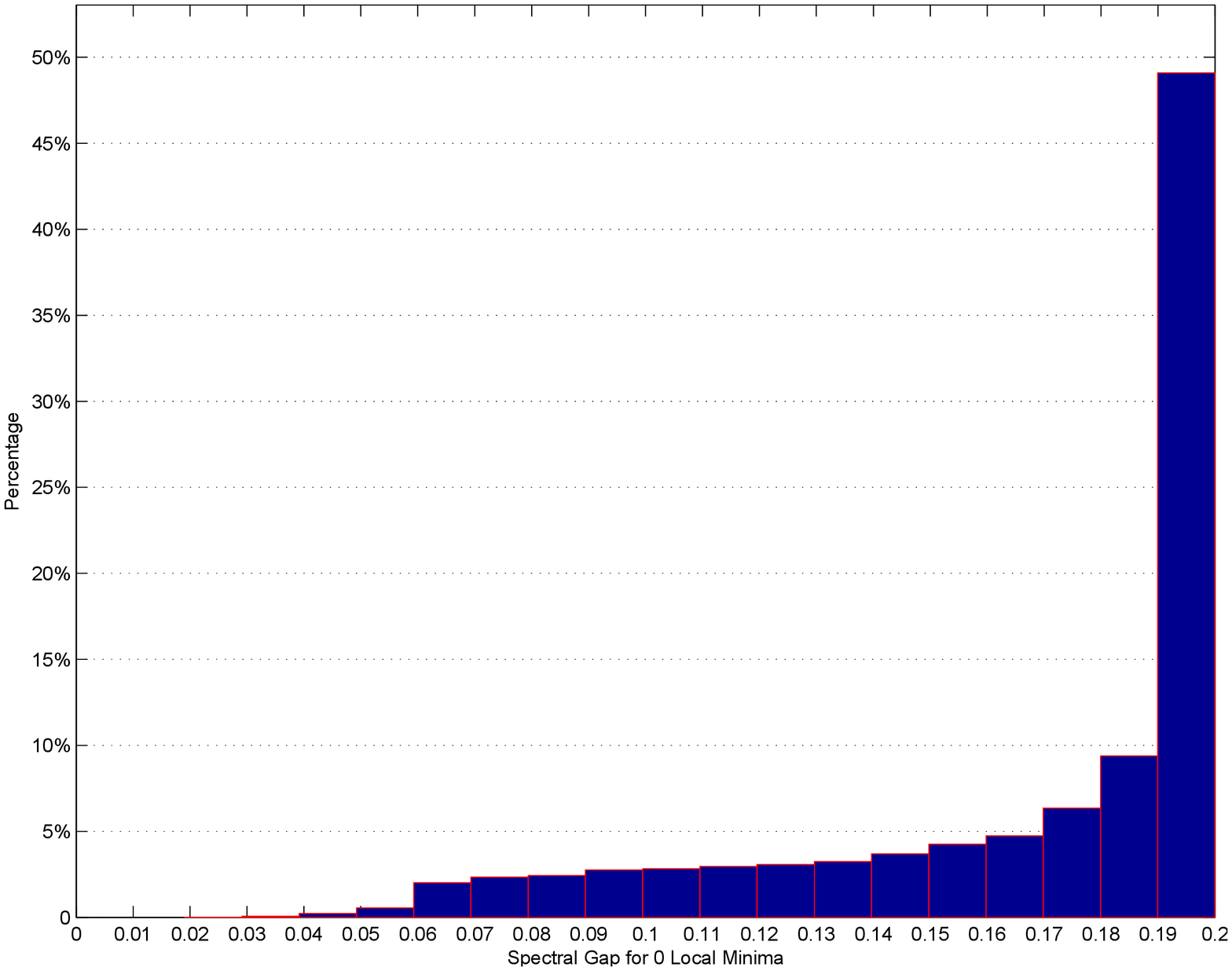}}
    \end{minipage}%
    \begin{minipage}{.5\linewidth}
        \centering
        \subfloat[]{\label{main:b}\includegraphics[scale=.20]{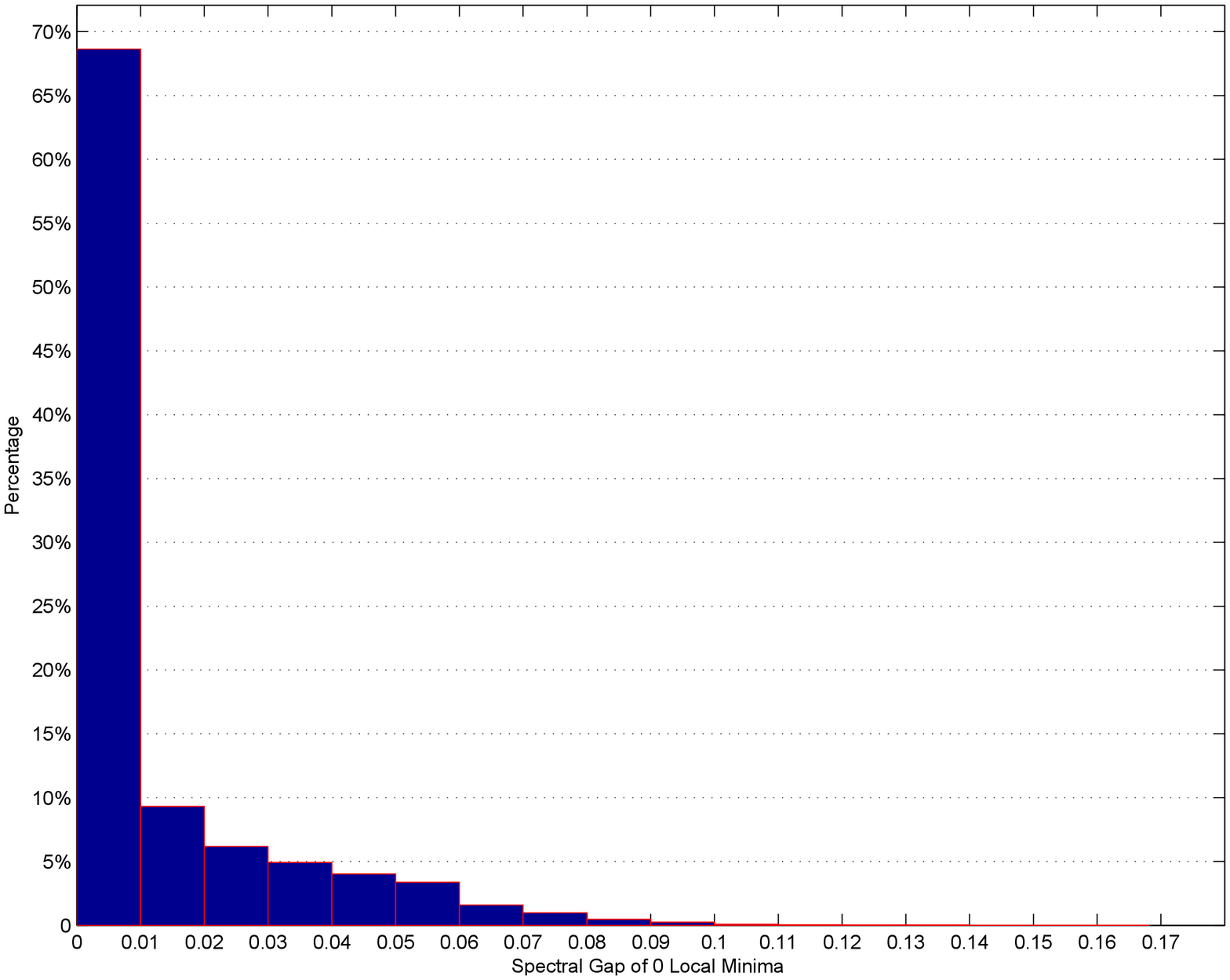}}
    \end{minipage}\par\medskip
    \begin{minipage}{.5\linewidth}
        \centering
        \subfloat[]{\label{main:c}\includegraphics[scale=.20]{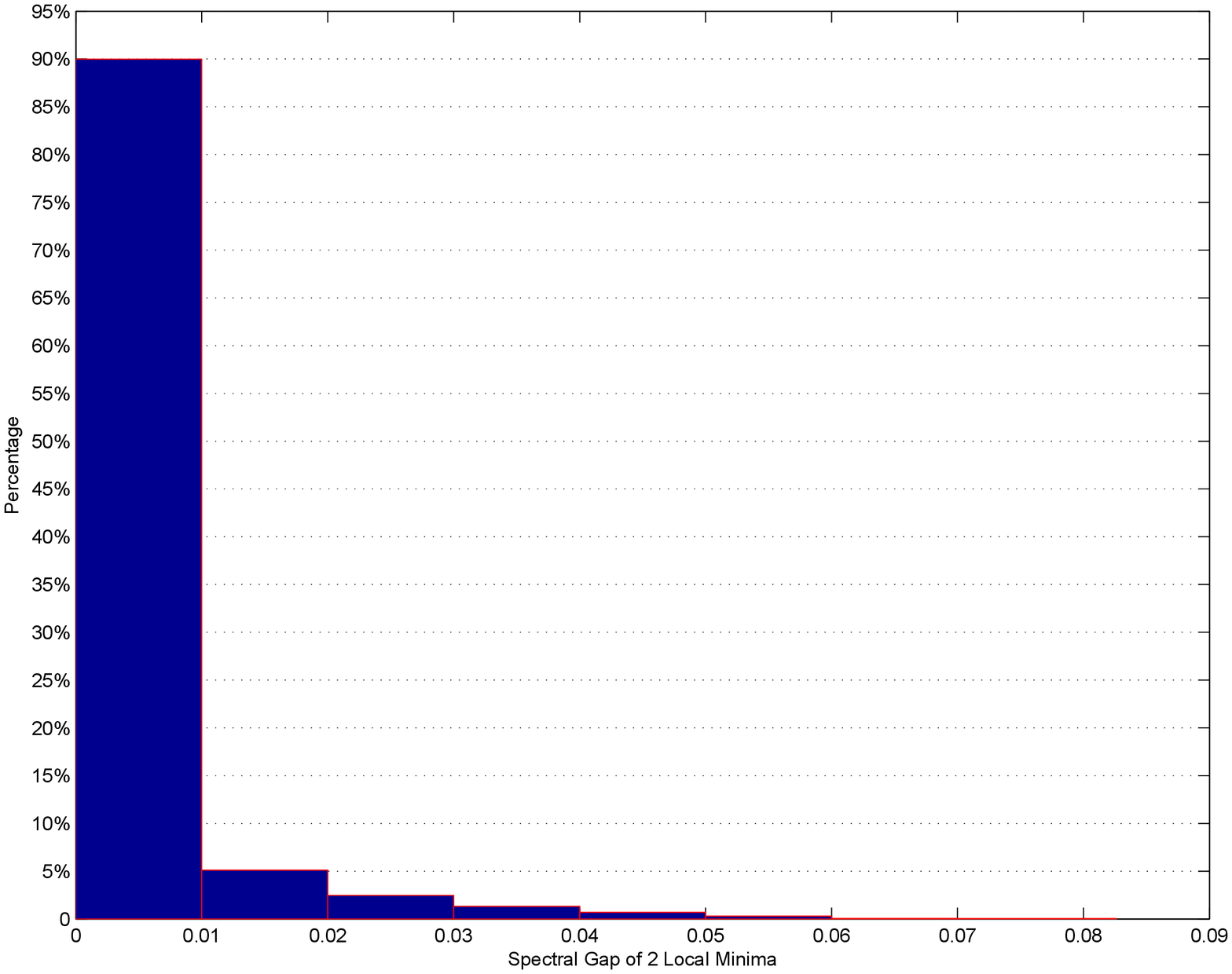}}
    \end{minipage}%
    \begin{minipage}{.5\linewidth}
        \centering
        \subfloat[]{\label{main:d}\includegraphics[scale=.20]{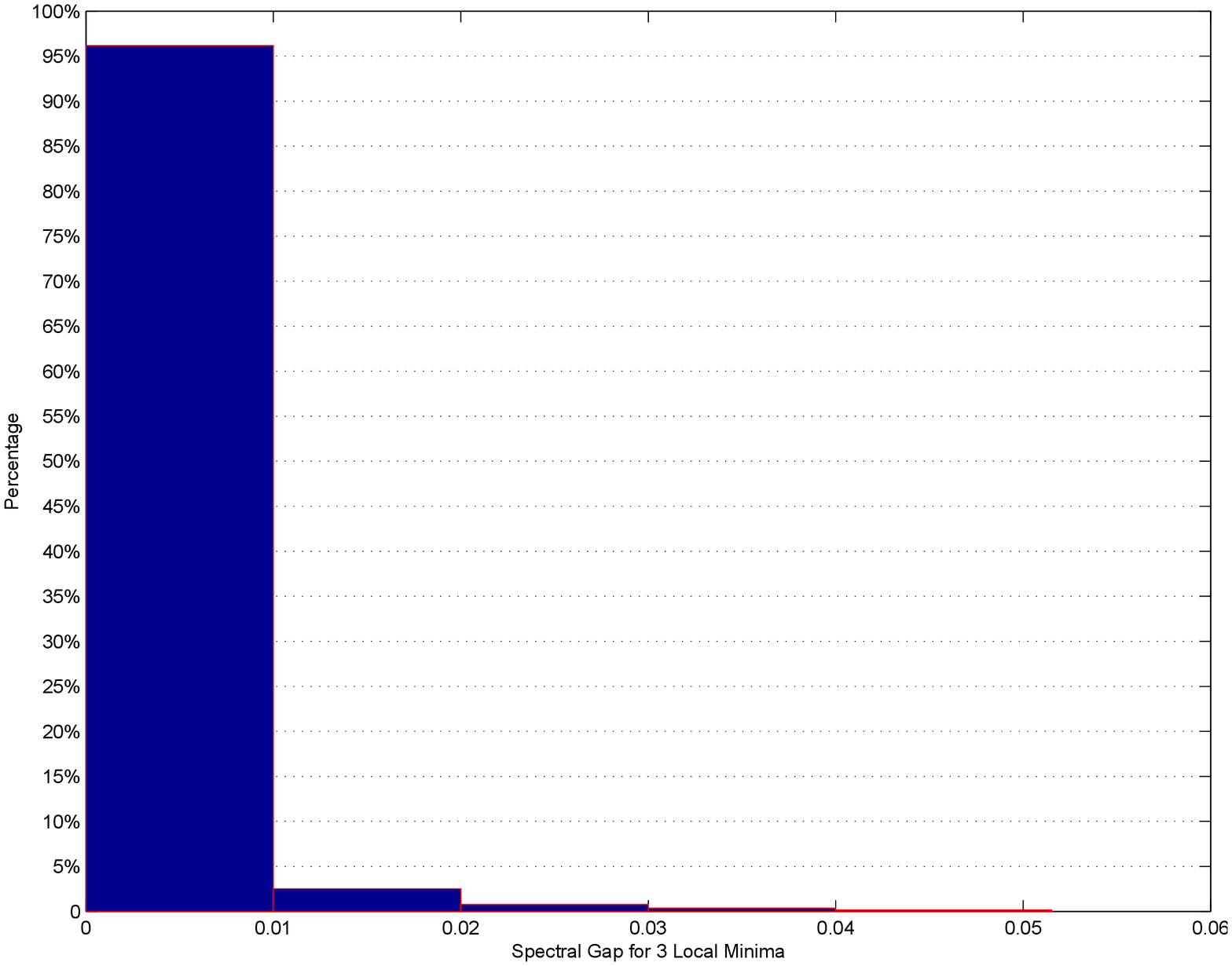}}
    \end{minipage}\par\medskip

\caption{Spectral gap with (a) 0 (b) 1 (c) 2 (d) 3 local minima}
\label{4}
\end{figure}

%
%
%
%

\section{Appendix}
\label{sec:Appendix}
\subsection{Proving Lemma \ref{Lem:Gaussian_integral}}
\label{subsec:Gaussian_integral}
{\bf Lemma \ref{Lem:Gaussian_integral}} (Gaussian Integral)
\textit{ Let $\vb$ and $\xb$ be independent
  Gaussian random vectors with distribution ${\cal N}(\zerob,\Ib_N)$
  each. Then 
\be
\Ec\left\{ e^{\eta\left(\|\vb+a \xb\|^2-\|\vb\|^2\right)} \right\} =
  \left(\frac{1}{1-2a^2\eta(1+2\eta)}\right)^{N/2} \ . 
\ee
}

\noindent {\bf Proof:}
In order to determine the expected value we compute the multivariate integral
\begin{subequations}
\begin{align}
& \Ec \left\{e^{\eta\left(\|\vb+a \xb\|^2-\|\vb\|^2\right)} \right\} \\
& = \int{\frac{d\xb d\vb}{(2\pi)^N}
e^{-\frac{1}{2} \ba{cc} \hspace{-1mm} \vb^T, & \hspace{-2mm} \xb^T \hspace{-1mm}\ea 
\ba{cc} \hspace{-1mm} \Ib_N & -2a\eta \Ib_N \hspace{-1mm} \\ \hspace{-1mm} -2a\eta \Ib_N & (1-2a^2\eta)\Ib_N \hspace{-1mm} \ea 
\ba{c} \hspace{-1mm} \vb \hspace{-1mm} \\ \hspace{-1mm} \xb \hspace{-1mm} \ea }} \\
& = \frac{1}{\mbox{det}^{1/2}
\ba{cc} \Ib_N & -2a\eta \Ib_N \\ -2a\eta \Ib_N & (1-2a^2\eta)\Ib_N \ea} \\
& = \frac{1}{\mbox{det}^{N/2}
\ba{cc} 1 & -2a\eta \\ -2a\eta & 1-2a^2\eta \ea} \\
& = \left(\frac{1}{1-2a^2\eta(1+2\eta)}\right)^{N/2} \ .
\end{align}
\end{subequations}
Thus, Lemma \ref{Lem:Gaussian_integral} has hereby been proved.
\qd

\subsection{Proving Lemma \ref{thm:manylocalminima}}
\label{subsec:proofmanylocalminima}
\begin{proof}
Let $N$ be an even integer. Consider a matrix whose first $\frac{N}{2}$ columns ${\hb}_{i}$, $1\leq i \leq \frac{N}{2}$ have unit norms and are orthogonal to each other. For the other  $\frac{N}{2}$ columns ${\hb}_{i}$, $\frac{N}{2} +1\leq i \leq N$, ${\hb}_{i}=-(1+\epsilon){\hb}_{i-\frac{n}{2}}$, where $\epsilon$ is a sufficiently small positive number ($\epsilon<1$). We also let $\yb={{\acute{\Hb}}} (-\mathbf{{1}})$, where $\mathbf{{1}}$ is an all-$1$ vector. So $-\mathbf{{1}}$ is a globally minimum point for this ILS problem.

Consider all those vectors ${\tilde{\xb}}'$ which, for any $1\leq i\leq \frac{N}{2}$, its $i$-th element and  $i+\frac{N}{2}$-th element are either simultaneously $+1$ or simultaneously $-1$. When $\epsilon$ is smaller than $1$, we claim that any such a vector except the all $-1$ vector ${\tilde{\xb}}$, is a local minimum, which shows that there are at least $2^{\frac{N}{2}}-1$ local minima.

 Assume that for a certain $1\leq i\leq \frac{N}{2}$, the $i$-th element and  $(i+\frac{N}{2})$-th element of ${\tilde{\xb}}'$ are simultaneously $-1$. Then if we change the $i$-th element to $+1$, $\|\yb-{{\acute{\Hb}}}{\tilde{\xb}}'\|^2$ increases by $4$; and if we change the $(i+\frac{N}{2})$-th element to $+1$, $\|\yb-{{\acute{\Hb}}}{\tilde{\xb}}'\|^2$ increases by $4(1+\epsilon)^2$. This is true because the $i$-th and $(i+\frac{N}{2})$-th columns are orthogonal to other $(N-2)$ columns.

Similarly, assume that for a certain $1\leq i\leq \frac{N}{2}$, the $i$-th element and  $(i+\frac{N}{2})$-th element of ${\tilde{\xb}}'$ are simultaneously $+1$. Then if we change the $i$-th element to $-1$, $\|\yb-{{\acute{\Hb}}}{\tilde{\xb}}'\|^2$ increases by $4(1+\epsilon)^2-4\epsilon^2$; and if we change the $(i+\frac{N}{2})$-th element to $-1$, $\|\yb-{{\acute{\Hb}}}{\tilde{\xb}}'\|^2$ increases by $4-4\epsilon^2$.
\end{proof}

\subsection{Proving Lemma \ref{thm:22indcolumns}}
\label{subsec:proofthm:22indcolumns}

\begin{proof}
When $\upsib=0$,  clearly ${\tilde{\xb}}=(-1,-1)$ is a global minimum point, not a local minimum point. It is also clear that ${\tilde{\xb}}=(-1,1)$ or ${\tilde{\xb}}=(1,-1)$ can not be a local minimum point since they are neighbors to the global minimum solution. So the only possible local minimum point is ${\tilde{\xb}}=(1,1)$.

From Lemma \ref{lemma:localcondition}, the corresponding necessary and sufficient condition is
\begin{equation*}
{\hb}_{1}^{T}{\hb}_{2} < -\frac{\|{\hb}_{1}\|^2}{2}=-\frac{\|{\hb}_{2}\|^2}{2}=-\frac{1}{2}.
\end{equation*}
This means the angle $\theta$ between the two 2-dimensional vectors ${\hb}_{1}$ and ${\hb}_{2}$ satisfy $\cos(\theta) <-\frac{1}{2}$. Since ${\hb}_{1}$ and ${\hb}_{2}$ are two independent uniform randomly sampled vector, the chance for that to happen is $\frac{\pi-\arccos{(-\frac{1}{2})}}{\pi}=\frac{1}{3}$.
\end{proof}

\subsection{Proving Lemma \ref{thm:22Gaussian}}
\label{subsec:proofthm:22Gaussian}
\begin{proof}
When $\upsib=0$,  clearly ${\tilde{\xb}}=(-1,-1)$ is a global minimum point, not a local minimum point. It is also clear that ${\tilde{\xb}}=(-1,1)$ or ${\tilde{\xb}}=(1,-1)$ can not be a local minimum point since they are neighbors to the global minimum solution. So the only possible local minimum point is ${\tilde{\xb}}=(1,1)$.

From Lemma \ref{lemma:localcondition}, the corresponding necessary and sufficient condition is
\begin{equation*}
{\hb}_{1}^{T}{\hb}_{2} < -\max\left\{ \frac{\|{\hb}_{1}\|^2}{2}, \frac{\|{\hb}_{2}\|^2}{2} \right\}.
\end{equation*}
This means the angle $\theta$ between the two 2-dimensional vectors ${\hb}_{1}$ and ${\hb}_{2}$ satisfy
\begin{equation*}
r_1 r_2 \cos(\theta) < -\frac{\max\left\{r_1^2, r_2^2\right\}}{2},
\end{equation*}
where $r_1$ and $r_2$ are respectively the $\ell_2$ norm of ${\hb}_{1}$ and ${\hb}_{2}$.

 Because the elements of ${\acute{\Hb}}$ are independent Gaussian random variables, $r_1$ and $r_2$ are thus independent random variables following the Rayleigh distribution
\begin{eqnarray*}
p(r_1)=r_1 e^{-\frac{r_1^2}{2}}, p(r_2)=r_2 e^{-\frac{r_2^2}{2}},
\end{eqnarray*}
while $\theta$ follows a uniform distribution over $[0, 2\pi)$

By symmetry, for $t\geq 1$,
\begin{eqnarray*}
&&P(\frac{\max\left\{r_1^2, r_2^2\right\}}{r_1 r_2}>t) \\
&=&2\int_{0}^{\infty} r_1 e^{-\frac{r_1^2}{2}}  \times {\int_{0}^{\frac{r_1}{t}} r_2 e^{-\frac{r_2^2}{2}} \,dr_2}                           \,dr_1\\
&=&2\int_{0}^{\infty} r_1 e^{-\frac{r_1^2}{2}}  \times {(1-e^{-\frac{r_1^2}{2}})} \,dr_1\\
&=&2(1-\int_{0}^{\infty} r_1 e^{-(\frac{1}{2}+\frac{1}{2t^2})r_1^2} \,dr_1)\\
&=&\frac{2}{t^2+1}.
\end{eqnarray*}

Since $\theta$ is an independent random variable satisfying  $\cos(\theta) < -\frac{\max\left\{r_1^2, r_2^2\right\}}{2r_1 r_2}$ and  $\cos(\theta) \geq -1$, the probability that ${\tilde{\xb}}=(+1,+1)$ is a local minimum is given by
\begin{eqnarray*}
&P=& \int_{1}^{2} (1-\frac{2}{t^2+1})' (1-\frac{\arccos(-\frac{t}{2})}{\pi}) \,dt\\
&=& \int_{1}^{2} \frac{4t}{(t^2+1)^2} (1-\frac{\arccos(-\frac{t}{2})}{\pi}) \,dt. \\
&=&\frac{1}{3}-\frac{1}{\sqrt{5}}+\frac{2\arctan(\sqrt{\frac{5}{3}})}{\sqrt{5}\pi},
\end{eqnarray*}
which is approximately $0.145696$.
\end{proof}

\bibliographystyle{IEEEtran}
\bibliography{refs}
\end{document}